\documentclass[letterpaper, 11pt]{article}
\usepackage{appendix}
\usepackage{caption}

\newcommand{\comment}[1]{}
 
\usepackage{graphicx}

\usepackage{subfigure}
\usepackage{amsmath}
\usepackage{amssymb}
\usepackage{mdwlist}

\usepackage{algorithm}

\usepackage{hyperref}

\usepackage{xspace}
\usepackage{setspace}

\usepackage{framed}
\usepackage{algpseudocode}
\usepackage{algorithmicx}

\newcommand{\pa}{p_{ua}}
\newcommand{\pc}{p_{ic}}

\newenvironment{proof}{\noindent {\bf Proof:}~}{\hspace*{\fill}\(\Box\)}

\newtheorem{theorem}{Theorem}

\newtheorem{definition}{Definition}

\def\noflash#1{\setbox0=\hbox{#1}\hbox to 1\wd0{\hfill}}

\begin{document}
\title{Characterizing and Adapting the Consistency-Latency Tradeoff in Distributed Key-value Stores
\footnote{\normalsize This work was supported in part by the following grants: NSF CNS 1409416, NSF CNS 1319527, NSF CCF 0964471, AFOSR/AFRL FA8750-11-2-0084, and a VMware Graduate Fellowship.}}

\author{Muntasir Raihan Rahman$^{1}$, Lewis Tseng$^{1,3}$, Son Nguyen$^{1}$, Indranil Gupta$^{1}$, and Nitin Vaidya$^{2,3}$\\~\\
 \normalsize $^1$ Department of Computer Science,\\
 \normalsize $^2$ Department of Electrical and Computer Engineering, 
 and\\ \normalsize $^3$ Coordinated Science Laboratory\\ \normalsize University of Illinois at Urbana-Champaign\\ \normalsize Email: \{mrahman2, ltseng3, nguyenb1, indy, nhv\}@illinois.edu}

\date{}
\maketitle

\begin{abstract}
{\normalsize

The CAP theorem is a fundamental result that applies to distributed storage systems. In this paper, we first present and prove two CAP-like impossibility theorems. To state these theorems, we present probabilistic models to characterize the three important elements of the CAP theorem: consistency (C), availability or latency (A), and partition tolerance (P). The theorems show the un-achievable envelope, i.e., which combinations of the parameters of the three models make them impossible to achieve together. Next, we present the design of a class of systems called PCAP that perform close to the envelope described by our theorems. In addition, these systems allow applications running on a single data-center to specify either a latency SLA or a consistency SLA. The PCAP systems automatically adapt, in real-time and under changing network conditions, to meet the SLA while optimizing the other C/A metric. We incorporate PCAP into two popular key-value stores -- Apache Cassandra and Riak. Our experiments with these two deployments, under realistic workloads, reveal that the PCAP system satisfactorily meets SLAs, and performs close to the achievable envelope. We also extend PCAP from a single data-center to multiple geo-distributed data-centers.

}
\end{abstract}

\thispagestyle{empty}
\newpage
\setcounter{page}{1}

\section{Introduction}
\label{s_intro}

Storage systems form the foundational platform for modern Internet services such as Web search, analytics, and social networking. Ever increasing user bases and massive data sets have forced users and applications to forgo conventional relational databases, and move towards a new class of scalable storage systems known as NoSQL key-value stores. Many of these distributed key-value stores (e.g., Cassandra~\cite{cassandra}, Riak~\cite{riak}, Dynamo~\cite{DHJ07}, Voldemort~\cite{voldemort}) support a simple GET/PUT interface for accessing and updating data items. The data items are replicated at multiple servers for fault tolerance. In addition, they offer a very weak notion of consistency known as eventual consistency~\cite{V09,ectoday}, which roughly speaking, says that if no further updates are sent to a given data item, all replicas will eventually hold the same value.

 These key-value stores are preferred by applications for whom eventual consistency suffices, but where high availability and low latency (i.e., fast reads and writes~\cite{abadi:cap}) are paramount. Latency is a critical metric for such cloud services because latency is correlated to user satisfaction -- for instance, a 500 ms increase in latency for operations at Google.com can cause a 20\% drop in revenue \cite{activo}. At Amazon, this translates to a \$6M yearly loss per added millisecond of latency \cite{U09}. This correlation between delay and lost retention is fundamentally human. Humans suffer from a phenomenon called \emph{user cognitive drift}, wherein if more than a second (or so) elapses between clicking on something and receiving a response, the user's mind is already elsewhere. 

At the same time, clients in such applications expect freshness, i.e., data returned by a read to a key should come from the latest writes done to that key by any client. For instance, Netflix uses Cassandra to track positions in each video~\cite{C13}, and freshness of data translates to accurate tracking and user satisfaction. 
This implies that clients care about a {\it time-based} notion of data freshness. Thus, this paper focuses on consistency based on the notion of data freshness (as defined later).


The CAP theorem was proposed by Eric Brewer~\cite{B00,B10}, and later formally proved by Gilbert and Lynch~\cite{GL10,lg:cap}. It essentially states that a system can choose at most two of three desirable properties: Consistency (C), Availability (A), and Partition tolerance (P). Recently, Abadi~\cite{abadi:cap} proposed to study the consistency-latency tradeoff, and unified the tradeoff with the CAP theorem. The unified result is called PACELC. It states that when a network partition occurs, one needs to choose between Availability and Consistency, otherwise the choice is between Latency and Consistency. We focus on the latter tradeoff as it is the common case. These prior results provided qualitative characterization of the tradeoff between consistency and availability/latency, while we provide a \textit{quantitative} characterization of the tradeoff.


Concretely, traditional CAP literature tends to focus on situations where ``hard'' network partitions occur and the designer has to choose between C or A, e.g., in geo-distributed data-centers. 
However, individual data-centers themselves suffer far more frequently from {``soft'' partitions~\cite{dean2009}}, arising from periods of elevated message delays or loss rates (i.e., the ``otherwise" part of PACELC) within a data-center. Neither the original CAP theorem nor the existing work on consistency in key-value stores~\cite{pbsj,DHJ07,GBK11,WHCL15,LPC12,cops:sosp11,Eiger,cfrdt,pileus,V09} address such soft partitions for a single data-center. 

In this paper we state and prove two CAP-like impossibility theorems. To state these theorems, we present probabilistic\footnote{By probabilistic, we mean the behavior is statistical over a long time period.} models to characterize the three important elements: soft partition,  latency requirements, and consistency requirements.  All our models take timeliness into account. Our latency model specifies soft bounds on operation latencies, as might be provided by the application in an SLA (Service Level Agreement). Our consistency model captures the notion of data freshness returned by read operations. Our partition model describes propagation delays in the underlying network. The resulting theorems show the un-achievable envelope, i.e., which combinations of the parameters in these three models (partition, latency, consistency) make them impossible to achieve together. Note that the focus of the paper is neither defining a new consistency model nor comparing different types of consistency models. Instead, we are interested in the un-achievable envelope of the three important elements and measuring how close a system can perform to this envelop.

Next, we describe the design of a class of systems called PCAP (short for Probabilistic CAP) that perform close to the envelope described by our theorems. In addition, these systems allow applications running inside a single data-center to specify either a probabilistic latency SLA or a probabilistic consistency SLA. Given a probabilistic latency SLA,  PCAP's adaptive techniques meet the specified operational latency requirement, while optimizing the consistency achieved. Similarly, given a probabilistic consistency SLA, PCAP meets the consistency requirement while optimizing operational latency. 
PCAP does so under real and continuously changing network conditions. There are known use cases that would benefit from an latency SLA -- these include the Netflix video tracking application~\cite{C13}, online advertising~\cite{rtad}, and shopping cart applications~\cite{pileus} -- each of these needs fast response times but is willing to tolerate some staleness. A known use case for consistency SLA is a Web search application~\cite{pileus}, which desires search results with bounded staleness but would like to minimize the response time. While the PCAP system can be used with a variety of consistency and latency models (like PBS~\cite{pbsj}), we use our PCAP models for concreteness.


We have integrated our PCAP system into two key-value stores -- Apache Cassandra~\cite{cassandra} and Riak~\cite{riak}. Our experiments with these two deployments, using YCSB~\cite{ycsb} benchmarks, reveal that PCAP systems satisfactorily meets a latency SLA (or consistency SLA), optimize the consistency metric (respectively latency metric), perform reasonably close to the envelope described by our theorems, and scale well. 



We also extend PCAP from a single data-center to multiple geo-distributed data-centers. The key contribution of our second system (which we call GeoPCAP) is a set of rules for composing probabilistic consistency/latency models from across multiple data-centers in order to derive the global consistency-latency tradeoff behavior. Realistic simulations demonstrate that GeoPCAP can satisfactorily meet consistency/latency SLAs for applications interacting with multiple data-centers, while optimizing the other metric.

\section{Consistency-Latency Tradeoff}
\label{s_pcap}



We consider a key-value store system which provides a read/write API over an asynchronous distributed message-passing network. The system consists of clients and servers, in which, servers are responsible for replicating the data (or read/write object) and ensuring the specified consistency requirements, and clients can invoke a write (or read) operation that stores (or retrieves) some value of the specified key by contacting server(s). Specifically, in the system, data can be propagated from a writer client to multiple servers by a replication mechanism or background mechanism such as read repair~\cite{DHJ07}, and the data stored at servers can later be read by clients. There may be multiple versions of the data corresponding to the same key, and the exact value to be read by reader clients depends on how the system ensures the consistency requirements. Note that as addressed earlier, we define consistency based on freshness of the value returned by read operations (defined below).
We first present our probabilistic models for soft partition, latency and consistency. Then we present our impossibility results. These results only hold for a single data-center. Later in Section~\ref{sec:gpcap} we deal with the multiple data-center case. 


\subsection{Models}
\label{model} 

To capture consistency, we defined a new notion called {\it $t$-freshness}, which is a form of eventual consistency. 
Consider a single key (or read/write object) being read and written concurrently by multiple clients. An operation $O$  (read or write) has a start time $\tau_{start}(O)$ when the client issues $O$, and a finish time $\tau_{finish}(O)$ when the client receives an answer (for a read) or an acknowledgment (for a write). The write operation ends when the client receives an acknowledgment from the server. The value of a write operation can be reflected on the server side (i.e., visible to other clients) any time after the write starts. For clarity of our presentation, we assume that all write operations end in this paper, which is reasonable given client retries. Note that the written value can still propagate to other servers after the write ends by the background mechanism.
We assume that at time $0$ (initial time), the key has a default value.

\begin{definition}
\label{def:fresh}
\noindent{\bf $t$-freshness and $t$-staleness:}
A read operation $R$ is said to be \underline{$t$-fresh} if and only if $R$ returns a value written by any write operation that starts at or after time $\tau_{fresh}(R, t)$, which is defined below:

\begin{enumerate}
\item If there is at least one write starting in the interval $[\tau_{start}(R)-t, \tau_{start}(R)]$: 
then $\tau_{fresh}(R, t) = \tau_{start}(R)-t$.

\item If there is no write starting in the interval $[\tau_{start}(R)-t, \tau_{start}(R)]$, then there are two cases:

\begin{enumerate}
\item No write starts before $R$ starts: then $\tau_{fresh}(R, t)=0$. 

\item Some write starts before $R$ starts: then $\tau_{fresh}(R, t)$ is the start time of the last write operation that starts before
$\tau_{start}(R)-t$.
\end{enumerate}

\end{enumerate}
A read that is not $t$-fresh is said to be \underline{$t$-stale}.
\end{definition}

Note that the above characterization of $t_{fresh}(R,t)$ only depends
on {\em start times} of operations.


\begin{figure}[htbp]
	\centering
		\includegraphics[width=0.45\textwidth]{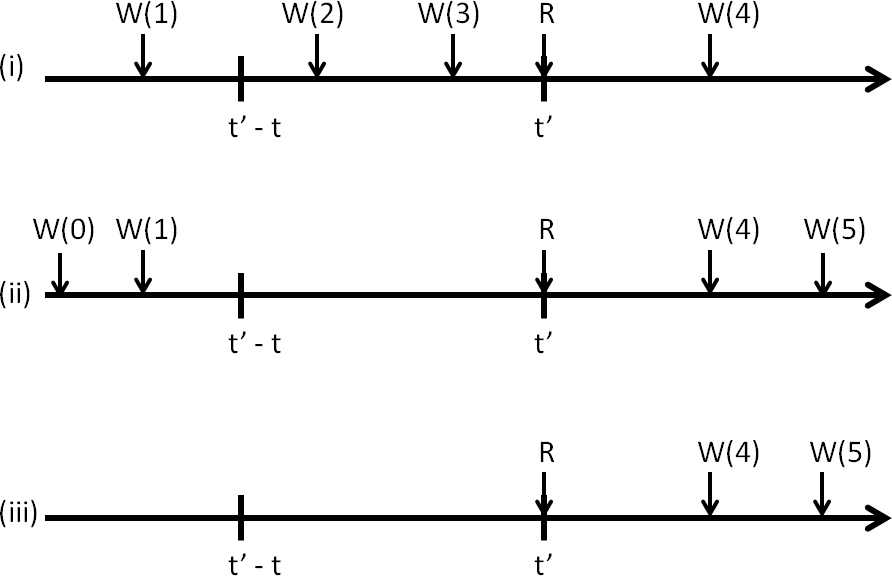}
		\caption{\it Examples illustrating Definition~\ref{def:fresh}. Only start times of each operation are shown.}
	\label{fig:fresh}
\end{figure}

Fig.~\ref{fig:fresh} shows three examples for $t$-freshness. The figure shows the times at which several read and write operations are issued (the time when operations complete are not shown in the figure). $W(x)$ in the figure denotes a write operation with a value $x$. Note that our definition of {\em $t$-freshness} allows a read to return a value that is written by a write issued after the read is issued. In Fig.~\ref{fig:fresh}(i), 
 $\tau_{fresh}(R, t) = \tau_{start}(R)-t = t'-t$;
therefore, $R$ is  $t$-fresh if it returns $2, 3$ or $4$. In Fig.~\ref{fig:fresh}(ii), $\tau_{fresh}(R, t) = \tau_{start}(W(1))$; therefore, $R$ is  $t$-fresh if it returns $1, 4$ or $5$. In Fig.~\ref{fig:fresh}(iii), $\tau_{fresh}(R, t) = 0$; therefore, $R$ is  $t$-fresh if it returns $4, 5$ or the default.



\begin{definition}
\label{def:t-consistency}
\noindent{\bf Probabilistic Consistency:} A key-value store satisfies $(t_c, \pc)$-consistency\footnote{The subscripts $c$ and $ic$ stand for consistency and inconsistency, respectively.} if in \underline{any execution} of the system, the fraction of read operations satisfying $t_c$-freshness is at least $(1-\pc)$.
\end{definition}

\begin{equation*}
\boxed{
\begin{aligned}
\text{Intuitively, $\pc$ is the {\it likelihood of returning stale data},}\\
\text{ given the time-based freshness requirement $t_c$.}
\end{aligned}}
\end{equation*}
\vspace{.5mm}

 Two similar definitions have been proposed previously: (1)  $t$-visibility from the Probabilistically Bounded Staleness (PBS) work \cite{pbsj}, and (2) $\Delta$-atomicity \cite{gls:fun}.
These two metrics do not require a read to return the latest write, but provide a time bound on the staleness of the data returned by the read. The main difference between $t$-freshness and these is that we consider the start time of write operations rather than the end time. This allows us to characterize consistency-latency tradeoff more precisely. While we prefer $t$-freshness, our PCAP system (Section~\ref{s_adaptive}) is modular and could use instead $t$-visibility or $\Delta$-atomicity for estimating data freshness.

{\ As noted earlier, our focus is not comparing different consistency models, nor achieving linearizability. We are interested in the un-achievable envelope of soft partition, latency requirements, and consistency requirements. Traditional consistency models like linearizability can be achieved by delaying the effect of a write. On the contrary, the achievability of $t$-freshness closely ties to the latency of read operations and underlying network behavior as discussed later. In other words, $t$-freshness by itself is not a complete definition.} 

\subsubsection{Use case for $\mathit{t-freshness}$}
Consider a bidding application (e.g., eBay), where everyone can post a bid, and we want every other participant to see posted bids as fast as possible. Assume that User 1 submits a bid, which is implemented as a write request (Figure~\ref{fig:fresh-motive}). User 2 requests to read the bid before the bid write process finishes. The same User 2 then waits a finite amount of time after the bid write completes and submits another read request.  Both of these read operations must reflect User 1's bid, whereas $t$-visibility only reflects the write in User 2's second read (with suitable choice of $t$). The bid write request duration can include time to send back an acknowledgment to the client, even after the bid has committed (on the servers). A client may not want to wait that long to see a submitted bid. This is especially true when the auction is near the end.



\begin{figure}[htbp]
	\centering
		\includegraphics[width=0.7\textwidth]{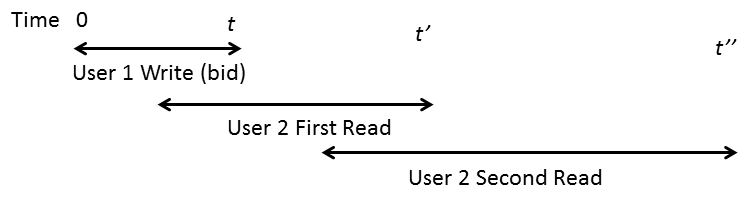}
		\caption{\it Example motivating use of Definition~\ref{def:t-consistency}.}
	\label{fig:fresh-motive}
\end{figure}

We define our probabilistic notion of latency as follows:

\begin{definition}
\label{def:t-latency}
\noindent { {{\bf $t$-latency:}} A read operation $R$ is said to satisfy \underline{$t$-latency} if and only if it completes within $t$ time units of its start time. }
\end{definition}


\begin{definition}
\label{def:t-availability}
\noindent {\bf Probabilistic Latency: } A key-value store satisfies $(t_a, \pa)$-latency\footnote{The subscripts $a$ and $ua$ stand for availability and unavailability, respectively.} if in \underline{any execution} of the system, the fraction of $t_a$-latency read operations is at least $(1-\pa)$.
\end{definition}



\begin{equation*}
\boxed{
\begin{aligned}
\text{Intuitively, given response time requirement $t_a$, }\\
\text{$\pa$ is the {\it likelihood of a read violating the $t_a$}.}
\end{aligned}}
\end{equation*}
\vspace{.3mm}

Finally, we capture the concept of a \emph{soft partition} of the network by defining a probabilistic partition model. {In this section, we assume that the partition model for the network does not change over time. (Later, our implementation and experiments in Section~\ref{exp} will measure the effect of time-varying partition models.)
 
In a key-value store, data can propagate from one client to another via the other servers using different approaches. 
For instance, in Apache Cassandra~\cite{cassandra}, a write might go from a writer client to a coordinator server to a replica server, or from a replica server to another replica server in the form of read repair~\cite{DHJ07}. Our partition model captures the delay of all such propagation approaches. 

\begin{definition}
\label{def:p-partition}
\noindent{\bf Probabilistic Partition:} 

An execution is said to suffer $(t_p,\alpha)$-partition if the fraction $f$ of paths from one client to another client, via a server, which have latency higher than $t_p$, is such that $f\geq\alpha$.


\end{definition}	

Our delay model loosely describes the message delay caused by any underlying network behavior without relying on the assumptions on the implementation of the key-value store. We do not assume eventual delivery of messages. We neither define propagation delay for each message nor specify the propagation paths (or alternatively, the replication mechanisms). This is because we want to have general lower bounds that apply to all systems that satisfy our models.

\subsection{Impossibility Results}
\label{captheorem}

We now present two theorems that characterize the consistency-latency tradeoff in terms of our probabilistic models.

First, we consider the case when the client has tight expectations, i.e., the client expects all data to be fresh within a time bound, and all reads need to be answered within a time bound. 

\begin{theorem}
\label{thm:CAP2}
If $~t_c + t_a < t_p$, then it is impossible to implement a read/write data object under a $(t_p, 0)$-partition while achieving {\em $(t_c, 0)$-consistency}, and {\em $(t_a, 0)$-latency}, i.e., there \underline{exists an execution} such that these three properties cannot be satisfied simultaneously.
\end{theorem}


\begin{proof}
The proof is by contradiction. 
In  a system 
that satisfies all three properties in all executions, consider an execution with only two clients, a writer client and a reader client. There are  two operations: (i) the writer client issues a write $W$, and (ii) the reader client issues a read  $R$ at time $\tau_{start}(R) = \tau_{start}(W)+t_c$. Due to $(t_c, 0)$-consistency, the read $R$ must return the value from $W$. 

Let the delay of the write request $W$ be exactly $t_p$ units of time (this obeys $(t_p, 0)$-partition). Thus, the earliest time that $W$'s value can arrive at the reader client is $(\tau_{start}(W)+t_p)$. However, to satisfy $(t_a, 0)$-latency, the reader client must receive an answer by time $\tau_{start}(R)+t_a=\tau_{start}(W)+t_c+t_a$. However, this time is earlier than $(\tau_{start}(W)+t_p)$ because $~t_c + t_a < t_p$. Hence, the value returned by $W$ cannot satisfy $(t_c, 0)$-consistency. This is a contradiction.
\end{proof}


~

Essentially, the above theorem relates the clients' expectations of freshness ($t_c$) and latency ($t_a$) to the propagation delays ($t_p$). If client expectations are too stringent when the maximum propagation delay is large, then it may not be possible to guarantee both consistency and latency expectations. 

However, if we allow a fraction of the reads to return late (i.e., after $t_a$), or return $t_c$-stale values (i.e., when either $\pc$ or $\pa$ is non-zero), 
then it may be possible to satisfy the three properties together even if $t_c + t_a < t_p$. 
Hence, we consider non-zero $\pc, \pa$ and $\alpha$ in our second theorem. 

\begin{theorem}
\label{thm:CAP-p}
If $~t_c + t_a < t_p$, and $~\pa + \pc < \alpha$, then it is impossible to implement a read/write data object under a $(t_p, \alpha)$-partition while achieving{\em $~(t_c, \pc)$-consistency}, and {\em $~(t_a, \pa)$-latency}, i.e., there \underline{exists an execution} such that these three properties cannot be satisfied simultaneously.
\end{theorem}


\begin{proof}
The proof is by contradiction. In a system that satisfies all three properties in all executions, consider an execution with only two clients, a writer client and a reader client. The execution contains alternating pairs of write and read operations $W_1, R_1, W_2, R_2, \ldots, W_n, R_n$, such that:



\begin{enumerate}
\item Write $W_i$ starts at time $(t_c + t_a)\cdot(i-1)$, 

\item Read $R_i$ starts at time $(t_c+t_a)\cdot(i-1)+t_c$, and

\item Each write $W_i$ writes a distinct value $v_i$.
\end{enumerate}

By our definition of $(t_p, \alpha)$-partition, there are at least $n\cdot\alpha$ written values $v_j$'s that have propagation delay $> t_p$. By a similar argument as in the proof of Theorem \ref{thm:CAP2}, each of their corresponding reads $R_j$ are such that $R_j$ cannot both satisfy $t_c$-freshness and also return within $t_a$. That is, $R_j$ is either $t_c$-stale or returns later than $t_a$ after its start time. There are $n\cdot\alpha$ such reads $R_j$; let us call these ``bad" reads.

By definition, the set of reads $S$ that are $t_c$-stale, and the set of reads $A$ that return after $t_a$ are such that $|S| \leq n\cdot\pc$ and $|A| \leq n\cdot\pa$. Put together, these imply:

$n\cdot\alpha \leq |S \cup A| \leq |S| + |A| \leq n\cdot\pc + n\cdot \pa$. 

The first inequality arises because all the ``bad" reads are in $S \cup A$. But this inequality implies that $\alpha \leq \pa + \pc$, which violates our assumptions.
\end{proof}


\section{PCAP Key-value Stores}



\label{s_adaptive}
 Having formally specified the (un)achievable envelope of consistency-latency (Theorem~\ref{thm:CAP-p}), we now move our attention to designing systems that achieve performance close to this theoretical envelope. We also convert our probabilistic models for consistency and latency from Section~\ref{s_pcap} into SLAs, and show how to design adaptive key-value stores that satisfy such probabilistic SLAs inside a single data-center. We call such systems PCAP systems. So PCAP systems (1) can achieve performance close to the theoretical consistency-latency tradeoff envelope, and (2) can adapt to meet probabilistic consistency and latency SLAs inside a single data-center. Our PCAP systems can also alternatively be used with SLAs from PBS~\cite{pbsj} or Pileus~\cite{pileus,tuba}.


\paragraph{Assumptions about underlying Key-value Store} PCAP systems can be built on top of existing key-value stores. We make a few assumptions about such key-value stores. First, we assume that each key is replicated on multiple servers. Second, we assume the existence of a ``coordinator'' server that acts as a client proxy in the system, finds the replica locations for a key (e.g., using consistent hashing~\cite{chord}), forwards client queries to replicas, and finally relays replica responses to clients. Most key-value stores feature such a coordinator~\cite{cassandra,riak}. Third, we assume the existence of a background mechanism such as read repair~\cite{DHJ07} for reconciling divergent replicas. Finally, we assume that the clocks on each server in the system are synchronized using a protocol like NTP so that we can use global timestamps to detect stale data (most key-value stores running within a datacenter already require this assumption, e.g., to decide which updates are fresher). It should be noted that our impossibility results in Section~\ref{s_pcap} do not depend on the accuracy of the clock synchronization protocol. However the sensitivity of the protocol affects the ability of PCAP systems to adapt to network delays. For example, if the servers are synchronized to within 1~ms using NTP, then the PCAP system cannot react to network delays lower than 1~ms.


\paragraph{SLAs} We consider two scenarios, where the SLA specifies either: i) a probabilistic latency requirement, or ii) a probabilistic consistency requirement. In the former case, our adaptive system optimizes the probabilistic consistency while meeting the SLA requirement, whereas in the latter it optimizes probabilistic latency while meeting the SLA. These SLAs are probabilistic, in the sense that they give statistical guarantees to operations over a long duration.

\vspace{1mm}

A latency SLA (i) looks as follows:


\begin{framed}
\noindent {\bf Given: } Latency $SLA=<\pa^{sla}, t_a^{sla}, t_c^{sla}>$;\\
{\bf Ensure that: } The fraction $\pa$ of reads, whose finish and start times differ by more than $ t_a^{sla}$, is such that: $\pa$ stays below $\pa^{sla}$ ;\\ 
{\bf Minimize: } The fraction $\pc$ of reads which do not satisfy $t_c^{sla}$-freshness.
\end{framed}

{ This SLA is similar to latency SLAs used in industry today.} As an example, consider a shopping cart application~\cite{pileus} where the client requires that at most 10\% of the operations take longer than 300~ms, but wishes to minimize staleness. Such an application prefers latency over consistency. In our system, this requirement can be specified as the following PCAP latency SLA:\\ 
\indent $<\pa^{sla}, t_a^{sla}, t_c^{sla}> = <0.1, 300~ms, 0~ms>$.


A consistency SLA looks as follows:
\begin{framed}
\noindent {\bf Given: } Consistency $SLA=<\pc^{sla}, t_a^{sla}, t_c^{sla}>$;\\
{\bf Ensure that: } The fraction $\pc$ of reads that do not satisfy $t_c^{sla}$-freshness is such that: $\pc$ stays below $\pc^{sla}$ ;\\
{\bf Minimize: } The fraction $\pa$ of reads whose finish and start times differ by more than $ t_a^{sla}$.
\end{framed}

Note that as mentioned earlier, consistency is defined based on freshness of the value returned by read operations. As an example, consider a web search application that wants to ensure no more than 10\% of search results return data that is over 500~ms old, but wishes to minimize the fraction of operations taking longer than 100~ms~\cite{pileus}. Such an application prefers consistency over latency. This requirement can be specified as the following PCAP consistency SLA:\\ 
\indent $<\pc^{sla}, t_a^{sla}, t_c^{sla}> = <0.10, 500~ms, 100~ms>$.



Our PCAP system can leverage three control knobs to meet these SLAs: 1)~{\it read delay}, 2)~{\it read repair rate}, and 3)~{\it consistency level}. The last two of these are present in most key-value stores. The first (read delay) has been discussed in previous literature~\cite{rdamazon,pbsj,Fan:2015:UCC:2752939.2752949,GW2014,Zawirski:2015:WFR:2814576.2814733}.  


\subsection{Control Knobs}
\label{knob}

\begin{figure}
{\small
\begin{center}
\begin{tabular}{|l|c|r|} \hline
Increased Knob	& Latency 	& Consistency\\ \hline \hline
Read Delay 		& Degrades	& Improves \\
Read Repair Rate 	& Unaffected	& Improves \\
Consistency Level 	& Degrades	& Improves \\ \hline
\end{tabular}
\end{center}
}
\caption{\it  {Effect of Various Control Knobs}.}
\label{table:knobs}
\end{figure}

Table~\ref{table:knobs} shows the effect of our three control knobs on latency and consistency. We discuss each of these knobs and explain the entries in the table. 


The knobs  of Table~\ref{table:knobs} are all directly or indirectly applicable to the read path in the key-value store.  As an example, the knobs pertaining to the Cassandra query path are shown in  Fig.~\ref{fig:knobs}, which shows the four major steps involved in answering a read query from a front-end to the key-value store cluster: (1) Client sends a read query for a key to a coordinator server in the key-value store cluster; (2) Coordinator forwards the query to one or more replicas holding the key; (3) Response is sent from replica(s) to coordinator; (4) Coordinator forwards response with highest timestamp to client; (5) Coordinator does {\it read repair} by updating replicas, which had returned older values, by sending them the freshest timestamp value for the key. Step (5) is usually performed in the background.
\begin{figure}[htbp]
	\centering
		\includegraphics[width=0.45\textwidth]{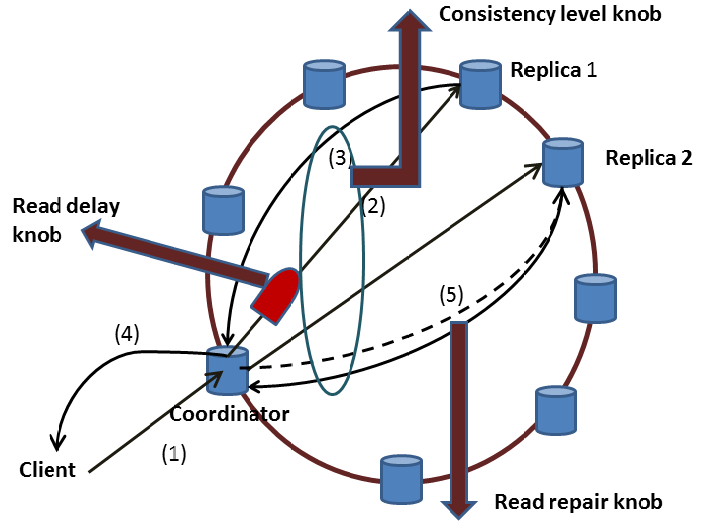}
	\caption{\it Cassandra Read Path and PCAP Control Knobs.}
	\label{fig:knobs}
\end{figure}

A {\it read delay} involves the coordinator artificially delaying the read query for a specified duration of time before forwarding it to the replicas. i.e., between step (1) and step (2). This gives the system some time to converge after previous writes. Increasing the value of read delay improves consistency (lowers $\pc$) and { degrades} latency (increases $\pa$). Decreasing read delay achieves the reverse. Read delay is an attractive knob because: 1)  it does not interfere with client specified parameters (e.g., consistency level in Cassandra~\cite{cassie-cl}), and 2)  it can take any non-negative continuous value instead of only discrete values allowed by consistency levels. Our PCAP system inserts read delays only when it is needed to satisfy the specified SLA.


However, read delay cannot be negative, as one cannot speed up a query and send it back in time. This brings us to our second knob: read repair rate. Read repair was depicted as distinct step (5) in our outline of Fig.~\ref{fig:knobs}, and is typically performed in the background. The coordinator maintains a buffer of recent reads where some of the replicas returned older values along with the associated freshest value. It periodically picks an element from this buffer and updates the appropriate replicas. { In key-value stores like Apache Cassandra and Riak, read repair rate is an accessible configuration parameter per column family}.

Our read repair rate knob is the probability with which a given read that returned stale replica values will be added to the read repair buffer. Thus, a read repair rate of 0 implies no read repair, and replicas will be updated only by subsequent writes. Read repair rate = 0.1 means the coordinator performs read repair for 10\% of the read requests.   

Increasing (respectively, decreasing) the read repair rate can improve (respectively degrade) consistency.  Since the read repair rate does not directly affect the read path (Step (5) described earlier, is performed in the background), it does not affect latency. Table~\ref{table:knobs} summarizes this behavior.\footnote{ Although read repair rate does not affect latency directly, it introduces some background traffic and can impact propagation delay. While our model ignores such small impacts, our experiments reflect the net effect of the background traffic.} 

The third potential control knob is consistency level. Some key-value stores allow the client to specify, along with each read or write operation, how many replicas the coordinator should wait for (in step (3) of Fig.~\ref{fig:knobs}) before it sends the reply back in step (4). For instance, Cassandra offers consistency levels: \texttt{ONE}, \texttt{TWO}, \texttt{QUORUM}, \texttt{ALL}. As one increases consistency level from \texttt{ONE} to \texttt{ALL}, reads are delayed longer (latency decreases) while the possibility of returning the latest write rises (consistency increases). 

Our PCAP system relies primarily on read delay and repair rate as the control knobs.  Consistency level can be used as a control knob only for applications in which user expectations will not be violated, e.g., when reads do not specify a specific discrete consistency level. That is, if a read specifies a higher consistency level, it would be prohibitive for the PCAP system to degrade the consistency level as this may violate client expectations. Techniques like continuous partial quorums (CPQ)~\cite{GolabCPQ15}, and adaptive hybrid quorums~\cite{pbsf} fall in this category, and thus interfere with application/client expectations. Further, read delay and repair rate are {\it non-blocking} control knobs under replica failure, whereas consistency level is {\it blocking}. For example, if a Cassandra client sets consistency level to \texttt{QUORUM} with replication factor 3, then the coordinator will be blocked if two of the key's replicas are on failed nodes. On the other hand, under replica failures read repair rate does not affect operation latency, while read delay only delays reads by a maximum amount. 

\subsection{Selecting A Control Knob} 
\label{choice}
As the primary control knob, the PCAP system prefers read delay over read repair rate. This is because the former allows tuning both consistency and latency, while the latter affects only consistency. The only exception occurs when during the PCAP system adaptation process, a state is reached where consistency needs to be degraded (e.g., increase $\pc$ to be closer to the SLA) but the read delay value is already zero. Since read delay cannot be lowered further, in this instance the PCAP system switches to using the secondary knob of read repair rate, and starts decreasing this instead. 

Another reason why read repair rate is not a good choice for the primary knob is that it takes longer to estimate $\pc$ than for read delay. Because read repair rate is a probability, the system needs a larger number of samples (from the operation log) to accurately estimate the actual $\pc$ resulting from a given read repair rate. For example, in our experiments, we observe that the system needs to inject $k\ge 3000$ operations to obtain an accurate estimate of $\pc$, whereas only $k=100$ suffices for the read delay knob.

\vspace{1.5mm}
\subsection{PCAP Control Loop}
\label{control}

\begin{figure}
\begin{algorithmic}[1]
\Procedure 
{control}{$\mathcal{SLA}=<\pc^{sla},t_c^{sla},t_a^{sla}>, \epsilon$}
	\State $\pc^{sla'} := \pc^{sla}-\epsilon$;
	\State Select control\_knob; // {\bf (Sections~\ref{knob}, \ref{choice})}
	\State $inc := 1$;
	\State $dir = +1$;
	\While{(true)}
		\State Inject $k$ new operations (reads and writes) 
		\State \,\,\,into store;
		\State Collect log $\mathcal{L}$ of recent completed reads  
		\State \,\,\,  and writes (values, start and finish times);
		\State Use $\mathcal{L}$ to calculate
		\State \,\,\,  $\pc$ and $\pa$; // {\bf (Section~\ref{sec:complexity})}
		\State $new\_dir := (\pc > \pc^{sla'})?+1:-1$;
		\If{$new\_dir = dir$}
			\State $inc := inc*2$; // Multiplicative increase
			\If {$inc > MAX\_INC$}
				\State $inc := MAX\_INC$:
			\EndIf
		\Else
			\State $inc := 1$; // Reset to unit step
			\State $dir := new\_dir$;	// Change direction
		\EndIf
		\State $control\_knob := control\_knob + inc * dir$;
	\EndWhile	
\EndProcedure
\end{algorithmic}
\caption{\it Adaptive Control Loop for Consistency SLA.}
\label{alg:adapt}
\end{figure}

The PCAP control loop adaptively tunes control knobs to always meet the SLA under continuously changing network conditions. The control loop for consistency SLA is depicted in Fig.~\ref{alg:adapt}. The control loop for a latency SLA is analogous and is not shown.
 

This control loop runs at a standalone server called the PCAP Coordinator.\footnote{The PCAP Coordinator is a special server, and is different from Cassandra's use of a coordinator for clients to send reads and writes.} This server runs an infinite loop. In each iteration, the coordinator: i) injects $k$ operations into the store (line 6), ii) collects the log $\mathcal{L}$ for the $k$ recent operations in the system (line 8), iii) calculates $\pa, \pc$  (Section~\ref{sec:complexity}) from $\mathcal{L}$ (line 10), and iv) uses these to change the knob (lines 12-22). 

The behavior of the control loop in Fig.~\ref{alg:adapt} is such that the system will converge to ``around'' the specified SLA. Because our original latency (consistency) SLAs require $\pa$ ($\pc$) to stay below the SLA, we introduce a \emph{laxity} parameter $\epsilon$,  subtract $\epsilon$ from the target SLA, and treat this as the target SLA in the control loop. Concretely, given a target consistency SLA $<\pc^{sla}, t_a^{sla}, t_c^{sla}>$, where the goal is to control the fraction of stale reads to be under $\pc^{sla}$, we control the system such that $\pc$ quickly converges  around $\pc^{sla'} = \pc^{sla}-\epsilon$, and thus stay below $\pc^{sla}$. Small values of $\epsilon$ suffice to guarantee convergence (for instance, our experiments use $\epsilon\le0.05$).

We found that the naive approach of  changing the control knob by the smallest unit increment (e.g., always 1~ms changes in read delay) resulted in a long convergence time. Thus, we opted for a {\it multiplicative} approach (Fig.~\ref{alg:adapt}, lines 12-22) to ensure quick convergence. 


We explain the control loop via an example. For concreteness, suppose only the read delay knob (Section~\ref{knob}) is active in the system, and that the system  has a consistency SLA. Suppose $\pc$ is higher than $\pc^{sla'}$. The multiplicative-change strategy starts incrementing the read delay, initially starting with a unit step size (line 3). This step size is exponentially {\it increased} from one iteration to the next, thus multiplicatively increasing read delay (line 14). This continues until the measured $\pc$ goes just under $\pc^{sla'}$. At this point, the {\it new\_dir} variable changes sign (line 12), so the strategy reverses direction, and the step is reset to unit size (lines 19-20). In subsequent iterations, the read delay starts {\it decreasing} by the step size. Again, the step size is increased exponentially until $\pc$ just goes above $\pc^{sla'}$. Then its direction is reversed again, and this process continues similarly thereafter. Notice that (lines 12-14) from one iteration to the next, as long as $\pc$ continues to remain above (or below) $\pc^{sla'}$, we have that: i) the direction of movement does not change, and ii) exponential increase continues. At steady state, the control loop keeps changing direction with a unit step size ({ bounded oscillation}), and the metric stays converged under the SLA. Although advanced techniques such as \emph{time dampening} can further reduce oscillations, we decided to avoid them to minimize control loop tuning overheads. Later in Section~\ref{sec:gpcap}, we utilized control theoretic techniques for the control loop in geo-distributed settings to reduce excessive oscillations. 

In order to prevent large step sizes, we cap the maximum step size (line 15-17). { For our experiments, we do not allow read delay to exceed  10~ms, and the unit step size is set to 1~ms.} 

We preferred active measurement (whereby the PCAP Coordinator injects queries rather than passive due to two reasons: i) the active approach gives the PCAP Coordinator better control on convergence, thus convergence rate is more uniform over time, and ii) in the passive approach if the client operation rate were to become low, then either the PCAP Coordinator would need to inject more queries, or convergence would slow down. Nevertheless, in Section~\ref{sec:passive}, we show results using a passive measurement approach. Exploration of hybrid active-passive approaches based on an operation rate threshold could be an interesting direction. 

Overall our PCAP controller satisfies SASO (Stability, Accuracy, low Settling time, small Overshoot) control objectives~\cite{Hellerstein:2004:FCC:975344}. 

\subsection{Complexity of Computing $\pa$ and $\pc$}
\label{sec:complexity} 

We show that the computation of $\pa$ and $\pc$ (line 10, Fig.~\ref{alg:adapt}) is efficient. Suppose there are $r$ reads and $w$ writes in the log, thus log size $k=r+w$. Calculating $\pa$ makes a linear pass over the read operations, and compares the difference of their finish and start times with $t_a$. This takes $O(r)=O(k)$.

{$\pc$ is calculated  as follows. We first extract and sort all the writes {according to start timestamp}, inserting each write into a hash table under key {$<$object value, write key, write timestamp$>$}.
In a second pass over the read operations, we extract its matching write by using the hash table key (the third entry of the hash key is the same as the read's returned value timestamp). We also extract neighboring writes of this matching write in constant time (due to the sorting), and thus calculate $t_c$-freshness for each read. The first pass takes time $O(r + w + w\log w)$, while the second pass takes $O(r+w)$}. {The total time complexity to calculate $\pc$ is thus $O(r + w + w\log w)=O(k\log k$).

\section{PCAP for Geo-distributed Settings}
\label{sec:gpcap}

In this section we extend our PCAP system from a single data-center to multiple geo-distributed data-centers. We call this system GeoPCAP.

\subsection{System Model} Assume there are $n$ data-centers. Each data-center stores multiple replicas for each data-item. When a client application submits a query, the query is first forwarded to the data-center closest to the client. We call this data-center the \emph{local} data-center for the client. If the local data-center stores a replica of the queried data item, that replica might not have the latest value, since write operations at other data-centers could have updated the data item. Thus in our system model, the local data-center contacts one or more of other { remote} data-centers, to retrieve (possibly) fresher values for the data item. 




\subsection{Probabilistic Composition Rules}
Each data-center is running our PCAP-enabled key-value store. Each such PCAP instance defines per data-center probabilistic latency and consistency models (Section~\ref{s_pcap}). To obtain the global behavior, we need to compose these probabilistic consistency and latency/availability models across different data-centers. This is done by our composition rules.

The composition rules for merging independent latency/consistency models from data-centers check whether the SLAs are met by the composed system. Since single data-center PCAP systems define probabilistic latency and consistency models, our composition rules are also probabilistic in nature. However in reality, our composition rules do not require all data-centers to run PCAP-enabled key-value stores systems. As long as we can measure consistency and latency at each data-center, we can estimate the probabilistic models of consistency/latency at each data-center and use our composition rules to merge them. 

We consider two types of composition rules: (1) \texttt{QUICKEST} (Q),  where at-least one data-center (e.g., the local or closest remote data-center) satisfies client specified latency or freshness (consistency) guarantees; and (2) \texttt{ALL} (A), where all the data-centers must satisfy latency or freshness guarantees. These two are, respectively, generalizations of Apache Cassandra multi-data-center deployment~\cite{cassandra-geo} consistency levels (CL): \texttt{LOCAL\_QOURUM} and \texttt{EACH\_QUORUM}.

Compared to Section~\ref{s_pcap}, which analyzed the fraction of executions that satisfy a predicate (the proportional approach), in this section we use a simpler probabilistic approach. This is because although the proportional approach is more accurate, it is more intractable than the probabilistic model in the geo-distributed case.


\begin{figure}
{\small
\begin{center}
\begin{tabular}{|l|l|l|l|} \hline
Consistency/Latency/WAN	& Composition & $\forall j, \,\,t_a^j = t$? & Rule\\ \hline \hline
Latency 		& \texttt{QUICKEST}	& Y & $\pa^c(t) = \Pi_j \,\,\pa^j$, $\forall j, \,\,t_a^j = t$ \\ \hline
Latency 	& \texttt{QUICKEST}	& N & $\pa^c(min_{j}\,\,t_a^j) \geq \Pi_j \,\,\pa^j \geq \pa^c(max_{j}\,\,t_a^j)$, \\
&&&$min_{j}\,\,t_a^j \leq t_a^c \leq max_{j}\,\,t_a^j$ \\ \hline
Latency  	& \texttt{ALL}	& Y & $\pa^c(t) = 1 - \Pi_j\,\,(1- \pa^j)$, $\forall j, \,\,t_a^j = t$ \\ \hline
Latency & \texttt{ALL} & N & $\pa^c(min_{j}\,\,t_a^j) \geq 1 - \Pi_j\,\,(1- \pa^j) \geq \pa^c(max_{j}\,\,t_a^j)$, \\
&&&$min_{j}\,\,t_a^j \leq t_a^c \leq max_{j}\,\,t_a^j$ \\ \hline
Consistency 		& \texttt{QUICKEST}	& Y & $\pc^c(t) = \Pi_j \,\,\pc^j$, $\forall j, \,\,t_c^j = t$ \\ \hline
Consistency 	& \texttt{QUICKEST}	& N & $\pc^c(min_{j}\,\,t_c^j) \geq \Pi_j \,\,\pc^j \geq \pc^c(max_{j}\,\,t_c^j)$, \\
&&&$min_{j}\,\,t_c^j \leq t_c^c \leq max_{j}\,\,t_c^j$ \\ \hline
Consistency  	& \texttt{ALL}	& Y & $\pc^c(t) = 1 - \Pi_j\,\,(1- \pc^j)$, $\forall j, \,\,t_c^j = t$ \\ \hline
Consistency & \texttt{ALL} & N & $\pc^c(min_{j}\,\,t_c^j) \geq 1 - \Pi_j\,\,(1- \pc^j) \geq \pc^c(max_{j}\,\,t_c^j)$, \\
&&&$min_{j}\,\,t_c^j \leq t_c^c \leq max_{j}\,\,t_c^j$ \\ \hline
Consistency-WAN & N.~A.~ & N.~A.~ & $Pr[X+Y \geq t_c+t_p^G] \geq \pc\cdot\alpha^G$ \\ \hline
Latency-WAN & N.~A.~ & N.~A.~ & $Pr[X+Y \geq t_a+t_p^G] \geq \pa\cdot\alpha^G$ \\ \hline

\end{tabular}
\end{center}
}
\caption{\it  {GeoPCAP Composition Rules}.}
\label{table:composition}
\end{figure}

Our probabilistic composition rules fall into three categories: (1) composing consistency models; (2) composing latency models; and (3) composing a wide-area-network (WAN) partition model with a data-center (consistency or latency) model. The rules are summarized in Figure~\ref{table:composition}, and we discuss them next.

\subsubsection{Composing latency models}
\label{sec:lat-comp}
 Assume there are $n$ data-centers storing the replica of a key with latency models $(t_a^1, \pa^1), (t_a^2, \pa^2), \ldots, (t_a^n, \pa^n)$. Let $\mathcal{C}^A$ denote the composed system. Let $(\pa^c,t_a^c)$ denote the latency model of the composed system $\mathcal{C}^A$. This indicates that the fraction of reads in $\mathcal{C}^A$ that complete within $t_a^c$ time units is at least $(1 - \pa^c)$. This is the latency SLA expected by clients. Let $\pa^c(t)$ denote the probability of missing deadline by $t$ time units in the composed model. Let $X_j$ denote the random variable measuring read latency in data center $j$. Let $E_j(t)$ denote the event that $X_j > t$. By definition we have that, $Pr[E_j(t_a^j)] = \pa^j$, and  $Pr[\bar{E}_j(t_a^j)] = 1 - \pa^j$. Let $f_j(t)$ denote the cumulative distribution function (CDF) for $X_j$. So by definition, $f_j(t_a^j) = Pr [X_j \leq t_a^j] = 1 - Pr [X_j > t_a^j]$. The following theorem articulates the probabilistic latency composition rules:
 
\begin{theorem}
\label{thm:latency-comp}
Let $n$ data-centers store the replica for a key with latency models $(t_a^1, \pa^1), (t_a^2, \pa^2), \ldots, (t_a^n, \pa^n)$. Let $\mathcal{C}^A$ denote the composed system with latency model $(\pa^c,t_a^c)$. Then for composition rule \texttt{QUICKEST} we have: 

\begin{eqnarray}
\label{eq:lat-comp-q}
\pa^c(min_{j}\,\,t_a^j) \geq \Pi_j \,\,\pa^j \geq \pa^c(max_{j}\,\,t_a^j), \\and \,\,\, min_{j}\,\,t_a^j \leq t_a^c \leq max_{j}\,\,t_a^j, \nonumber \\ where j \in\{1,\cdots,n\} \nonumber.
\end{eqnarray}

For composition rule \texttt{ALL}, 

\begin{eqnarray}
\label{eq:L}
\pa^c(min_{j}\,\,t_a^j) \geq 1 - \Pi_j\,\,(1- \pa^j) \geq \pa^c(max_{j}\,\,t_a^j), \\and \,\,\, min_{j}\,\,t_a^j \leq t_a^c \leq max_{j}\,\,t_a^j, \nonumber \\ where j \in\{1,\cdots,n\} \nonumber.
\end{eqnarray}

\end{theorem}

\begin{proof}
We outline the proof for composition rule \texttt{QUICKEST}. In \texttt{QUICKEST}, a latency deadline $t$ is violated in the composed model when all data-centers miss the $t$ deadline. This happens with probability $\pa^c(t)$ (by definition). We first prove a simpler Case 1, then the general version in Case 2. 

Case 1: Consider the simple case where all $t_a^j$ values are identical, i.e., $\forall j, t_a^j = t_a$: $\pa^c(t_a) = Pr [\cap_i E_i(t_a)] = \cap_i Pr[E_i(t_a)] = \Pi_i \pa^i$ (assuming independence across data-centers).

Case 2: 

Let, 

\begin{equation}
	t_a^i = min_j \,\,t_a^j
\end{equation}

Then,

\begin{equation}
\forall j, \,\,t_a^j \geq t_a^i
\end{equation}

Then, by definition of CDF function,

\begin{equation}
\forall j, \,\, f_j(t_a^i) \leq f_j(t_a^j)
\end{equation}

By definition,

\begin{equation}
\forall j, \,\, (Pr[X_j \leq t_a^i] \leq Pr[X_j\leq t_a^j])
\end{equation}

\begin{equation}
\forall j, \,\, (Pr[X_j > t_a^i] \geq Pr[X_j > t_a^j])
\end{equation}

Multiplying all,

\begin{equation}
\Pi_j\,\,Pr[X_j > t_a^i] \geq \Pi_j\,\,Pr[X_j > t_a^j]
\end{equation}

But this means,

\begin{equation}
\pa^c(t_a^i) \geq \Pi_j\,\,\pa^j
\end{equation}

\begin{equation}
\label{eq-proof-left}
\pa^c(min_j\,\,t_a^j) \geq \Pi_j\,\,\pa^j
\end{equation}

Similarly, let

\begin{equation}
	t_a^k = max_j \,\,t_a^j
\end{equation}

Then,

\begin{equation}
\forall j, \,\,t_a^k \geq t_a^j
\end{equation}

\begin{equation}
\forall j, \,\, (Pr[X_j > t_a^j] \geq Pr[X_j > t_a^k])
\end{equation}

\begin{equation}
\Pi_j\,\,Pr[X_j > t_a^j] \geq \Pi_j\,\,Pr[X_j > t_a^k]
\end{equation}

\begin{equation}
\Pi_j\,\,\pa^j \geq \pa^c(t_a^k)
\end{equation}

\begin{equation}
\label{eq-proof-right}
\Pi_j\,\,\pa^j \geq \pa^c(max_j\,\,t_a^j)
\end{equation}

Finally combining Equations~\ref{eq-proof-left}, and~\ref{eq-proof-right}, we get Equation~\ref{eq:lat-comp-q}.
 
  
  


The proof for composition rule \texttt{ALL} follows similarly. In this case, a latency deadline $t$ is satisfied when all data-centers satisfy the deadline. So a deadline miss in the composed model means at-least one data-center misses the deadline. The derivation of the composition rules are similar and we invite the reader to work them out to arrive at the equations depicted in Figure~\ref{table:composition}. 
\end{proof}

\subsubsection{Composing consistency models}
\label{sec:con-comp}

\begin{figure}[htbp]
	\centering
		\includegraphics[width=0.8\textwidth]{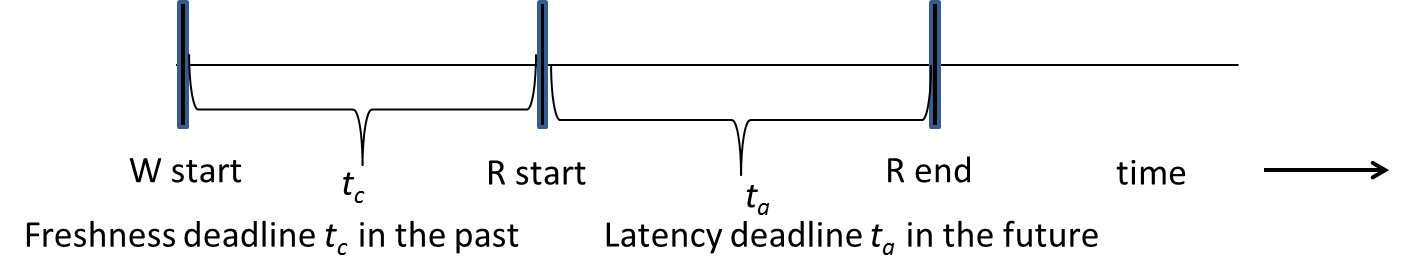}
	\caption{\it Symmetry of freshness and latency requirements.}
	\label{fig:symmetry}
\end{figure}

$t$-latency (Definition~\ref{def:t-latency}) and $t$-freshness (Definitions~\ref{def:fresh}) guarantees are time-symmetric (Figure~\ref{fig:symmetry}). While $t$-lateness can be considered a deadline in the future, $t$-freshness can be considered a deadline in the past. This means that for a given read, $t$-freshness constrains how old a read value can be. So the composition rules remain the same for consistency and availability. 

Thus the consistency composition rules can be obtained by substituting $\pa$ with $\pc$ and $t_a$ with $t_c$ in the latency composition rules (last 4 rows in Table~\ref{table:composition}).

This leads to the following theorem for consistency composition:

\begin{theorem}
\label{thm:con-comp}
Let $n$ data-centers store the replica for a key with consistency models $(t_c^1, \pc^1), (t_c^2, \pc^2), \ldots, (t_c^n, \pc^n)$. Let $\mathcal{C}^A$ denote the composed system with consistency model $(\pc^c,t_c^c)$. Then for composition rule \texttt{QUICKEST} we have: 

\begin{eqnarray}
\label{eq:con-comp-q}
\pc^c(min_{j}\,\,t_c^j) \geq \Pi_j \,\,\pc^j \geq \pc^c(max_{j}\,\,t_c^j), \\and \,\,\,  min_{j}\,\,t_c^j \leq t_c^c \leq max_{j}\,\,t_c^j, \nonumber \\ where j \in\{1,\cdots,n\} \nonumber.
\end{eqnarray}

For composition rule \texttt{ALL}, 

\begin{eqnarray}
\label{eq:con-comp-a}
\pc^c(min_{j}\,\,t_c^j) \geq 1 - \Pi_j\,\,(1- \pc^j) \geq \pc^c(max_{j}\,\,t_c^j), \\ and \,\,\,   min_{j}\,\,t_c^j \leq t_c^c \leq max_{j}\,\,t_c^j, \nonumber \\ where j \in\{1,\cdots,n\} \nonumber.
\end{eqnarray}
\end{theorem}

\subsubsection{Composing consistency/latency model with a WAN partition model}
\label{sec:wan-ac-comp}
All data-centers are connected to each other through a wide-area-network (WAN). We assume the WAN follows a partition model $(t_p^G, \alpha^G)$. This indicates that $\alpha^G$ fraction of messages passing through the WAN suffers a delay $> t_p^G$. Note that the WAN partition model is distinct from the per data-center partition model (Definition~\ref{def:p-partition}). 
Let $X$ denote the latency in a remote data-center, and $Y$ denote the WAN latency of a link connecting the local data-center to this remote data-center (with latency $X$). Then the total latency of the path to the remote data-center is $X+Y$.\footnote{We ignore the latency of the local data-center in this rule, since the local data-center latency is used in the latency composition rule (Section~\ref{sec:lat-comp}).}

\begin{equation}
\label{eq:wan-comp}
	Pr[X+Y\geq t_a+t_p^G] \geq (Pr[X\geq t_a]\cdot Pr[Y\geq t_p^G]) = \pa\cdot\alpha^G.
\end{equation}

Here we assume the WAN latency, and data-center latency distributions are independent. Note that Equation~\ref{eq:wan-comp} gives a lower bound of the probability. In practice we can estimate the probability by sampling both $X$ and $Y$, and estimating the number of times $(X+Y)$ exceeds $(t_a+t_p^G)$.

\subsection{Example}

The example in Figure~\ref{fig:comp-example} shows the composition rules in action.  In this example, there is one local data-center and 2 replica data-centers. Each data-center can hold multiple replicas of a data-item. First we compose each replica data-center latency model with the WAN partition model. Second we take the WAN-latency composed models for each data-center and compose them using the QUICKEST rule (Figure~\ref{table:composition}, bottom part).


\begin{figure}[htbp]
	\centering
		\includegraphics[width=0.85\textwidth]{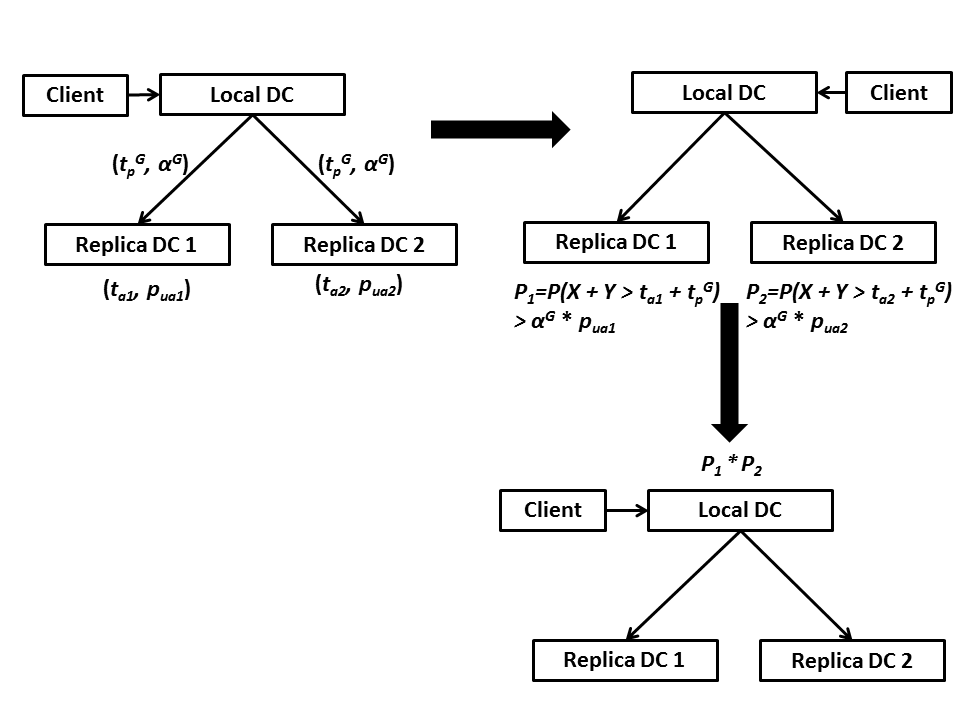}
	\caption{\it Example of composition rules in action.}
	\label{fig:comp-example}
\end{figure}

\subsection{GeoPCAP Control Knob}

We use a similar delay knob to meet the SLAs in a geo-distributed setting. We call this the \emph{geo-delay} knob and denote it as $\Delta$. The time delay $\Delta$ is the delay added at the local data-center to a read request received from a client before it is forwarded to the replica data-centers. $\Delta$ affects the consistency-latency trade-off in a manner similar to the read delay knob in a data-center (Section~\ref{knob}).  Increasing the knob tightens the deadline at each replica data-center, thus increasing per data-center latency ($\pa$). Similar to read delay (Figure~\ref{table:knobs}), increasing the geo delay knob improves consistency, since it gives each data-center time to commit latest writes. 

\subsection{GeoPCAP Control Loop}

Our GeoPCAP system uses a control loop depicted in Figure~\ref{alg:gadapt} for the Consistency SLA case using the \texttt{QUICKEST} composition rule. The control loops for the other three combinations (Consistency-\texttt{QUICKEST}, Latency-\texttt{ALL}, Latency-\texttt{QUICKEST}) are similar.

\begin{figure}
\begin{algorithmic}[1]
\Procedure 
{control}{$\mathcal{SLA}=<\pc^{sla},t_c^{sla},t_a^{sla}>$}
	\State Geo-delay $\Delta := 0$ 
	\State $E := 0$, $Error_{old} := 0$
	\State set $k_p$, $k_d$, $k_i$ for PID control (tuning)
	\State Let $(t_p^G,\alpha^G)$ be the WAN partition model
	\While{(true)}
		\For {each data-center $i$}
			\State Let $F_i$ denote the random freshness interval at $i$
			\State Let $L_i$ denote the random operation latency at $i$
			\State Let $W_i$ denote the WAN latency of the link to $i$
			\State Estimate $ \pc^i := Pr[F_i + W_i > t_c^{sla} + t_p^G + \Delta] $ // WAN composition (Section~\ref{sec:wan-ac-comp})
			\State Estimate $ \pa^i := Pr[L_i + W_i > t_a + t_p^G  - \Delta = t_a^{sla}]$
		\EndFor
		\State $\pc^c := \Pi_i\pc^i$, $\pa^c := \Pi_i \pa^i$ // Consistency/Latency composition (Sections~\ref{sec:lat-comp}~\ref{sec:con-comp})
		\State $Error := \pc^c - \pc^{sla}$
		\State $dE := Error - Error_{old}$
		\State $E := E + Error$
		\State $u := k_p \cdot Error + k_d \cdot dE + k_i \cdot E$
		\State $\Delta := \Delta + u$;
	\EndWhile	
\EndProcedure
\end{algorithmic}
\caption{\it Adaptive Control Loop for GeoPCAP Consistency SLA (\texttt{QUICKEST} Composition).}
\label{alg:gadapt}
\end{figure}

Initially, we opted to use the single data-center multiplicative control loop (Section~\ref{control}) for GeoPCAP. However, the multiplicative approach led to increased oscillations for the composed consistency ($\pc$) and latency ($\pa$) metrics in a geo-distributed setting. The multiplicative approach sufficed for the single data-center PCAP system, since the oscillations were bounded in steady-state. However, the increased oscillations in a geo-distributed setting prompted us to use a control theoretic approach for GeoCAP.

As a result, we use a PID control theory approach~\cite{AstHag95} for the GeoPCAP controller. The controller runs an infinite loop, so that it can react to network delay changes and meet SLAs. There is a tunable sleep time at the end of each iteration (1~sec in Section~\ref{sec:geo-pcap-eval} simulations). Initially the geo-delay $\Delta$ is set to zero. At each iteration of the loop, we use the composition rules to estimate $\pc^c(t)$, where $t = t_c^{sla} + t_p^G - \Delta$. We also keep track of composed $\pa^c()$ values. We then compute the error, as the difference between current composed $\pc$ and the SLA. Finally the geo-delay change is computed using the PID control law~\cite{AstHag95} as follows:

\begin{equation}
\label{eq:pid}
u = k_p\cdot Error(t) + k_d\cdot \frac{dError(t)}{dt} + k_i \cdot \int{Error(t)dt}
\end{equation}

Here, $k_p$, $k_d$, $k_i$ represent the proportional, differential, and integral gain factors for the PID controller respectively. There is a vast amount of literature on tuning these gain factors for different control systems~\cite{AstHag95}. Later in our experiments, we discuss how we set these factors to get SLA convergence. Finally at the end of the iteration, we increase $\Delta$ by $u$. Note that $u$ could be negative, if the metric is less than the SLA.

Note that for the single data-center PCAP system, we used a multiplicative control loop (Section~\ref{control}), which outperformed the unit step size policy. For GeoPCAP, we employ a PID control approach. PID is preferable to the multiplicative approach, since it guarantees fast convergence, and can reduce oscillation to arbitrarily small amounts. However PID's stability depends on proper tuning of the gain factors, which can result in high management overhead. On the other hand the multiplicative control loop has a single tuning factor (the multiplicative factor), so it is easier to manage.  Later in Section~\ref{sec:geo-pcap-eval} we experimentally compare the PID and multiplicative control approaches.




\section{Experiments}
\label{exp}

\subsection{Implementation Details}






In this section, we discuss how support for our consistency and latency SLAs can be easily incorporated into the Cassandra and Riak key-value stores (in a single data-center) via minimal changes.  
\subsubsection{PCAP Coordinator}

From Section~\ref{control}, recall that the PCAP Coordinator runs an infinite loop that continuously injects operations, collects logs ($k=100$ operations by default), calculates metrics, and changes the control knob. We implemented a modular PCAP Coordinator using Python  (around 100 LOC), which can be connected to any key-value store. 

We integrated PCAP into two popular NoSQL stores: Apache Cassandra~\cite{cassandra} and Riak~\cite{riak} -- each of these required changes to about 50 lines of original store code.



\subsubsection{Apache Cassandra}


First, we modified the Cassandra v1.2.4 to add read delay and read repair rate as control knobs. 
We changed the Cassandra Thrift interface so that it accepts read delay as an additional parameter. Incorporating the read delay into the read path required around 50 lines of Java code. 

Read repair rate is specified as a column family configuration parameter, and thus did not require any code changes. We used YCSB's Cassandra connector as the client, modified appropriately to talk with the clients and the PCAP Coordinator. 


\subsubsection{Riak}

We modified Riak v1.4.2 to add read delay and read repair as control knobs. Due to the unavailability of a YCSB Riak connector, we wrote a separate YCSB client for Riak from scratch (250 lines of Java code). We decided to use YCSB instead of existing Riak clients, since YCSB offers flexible workload choices that model real world key-value store workloads.
 
We introduced a new system-wide parameter for read delay, which was passed via the Riak http interface to the Riak coordinator which in turn applied it to all queries that it receives from clients. This required about 50 lines of Erlang code in Riak. Like Cassandra, Riak also has built-in support for controlling read repair rate.

\vspace*{-0.2cm}
\subsection{Experiment Setup}
\vspace*{-0.1cm}

Our experiments are in three stages: microbenchmarks for a single data-center (Section~\ref{sec:expt_micro}) and deployment experiments for a single data-center (Section~\ref{sec:expt_deployment}), and a realistic simulation for the geo-distributed setting (Section~\ref{sec:geo-pcap-eval}). We first discuss the experiments for a single data-center setting.

Our single data-center PCAP Cassandra system and our PCAP Riak system were each run with their default settings. We used YCSB v 0.1.4~\cite{ycsb1} to send operations to the store. { YCSB generates synthetic workloads for key-value stores and models real-world workload scenarios (e.g., Facebook photo storage workload). It has been used to benchmark many open-source and commercial key-value stores, and is the de facto benchmark for key-value stores~\cite{ycsb}.}

Each YCSB experiment consisted of a load phase, followed by a work phase.
Unless otherwise specified, we used the following YCSB parameters:
16 threads per YCSB instance, 2048~B values, and a read-heavy distribution (80\% reads). We had as many YCSB instances as the cluster size, one co-located at each server. The default key size was 10~B for Cassandra, and Riak. Both YCSB-Cassandra and YCSB-Riak connectors were used with the weakest quorum settings and 3 replicas per key. The default throughput was 1000 ops/s. All operations use a consistency level of \texttt{ONE}. 

Both PCAP systems were run in a cluster of 9 d710 Emulab servers \cite{emulab}, each with 4 core Xeon processors, 12~GB RAM, and 500~GB disks. 
The default network topology was a LAN (star topology), with 100~Mbps bandwidth and inter-server { round-trip} delay of 20~ms,  dynamically controlled using traffic shaping.




We used NTP to synchronize clocks within 1~ms. This is reasonable since we are limited to a single data-center. This clock skew can be made tighter by using atomic or GPS clocks~\cite{goo:span}. This synchronization is needed by the PCAP coordinator to compute the SLA metrics. 

\subsection{Microbenchmark Experiments (Single Data-center)}
\label{sec:expt_micro}


\subsubsection{Impact of Control Knobs on Consistency}

\label{sec:expt_knobs}
\begin{figure}[htbp]
	\centering
		\includegraphics[width=0.4\textwidth]{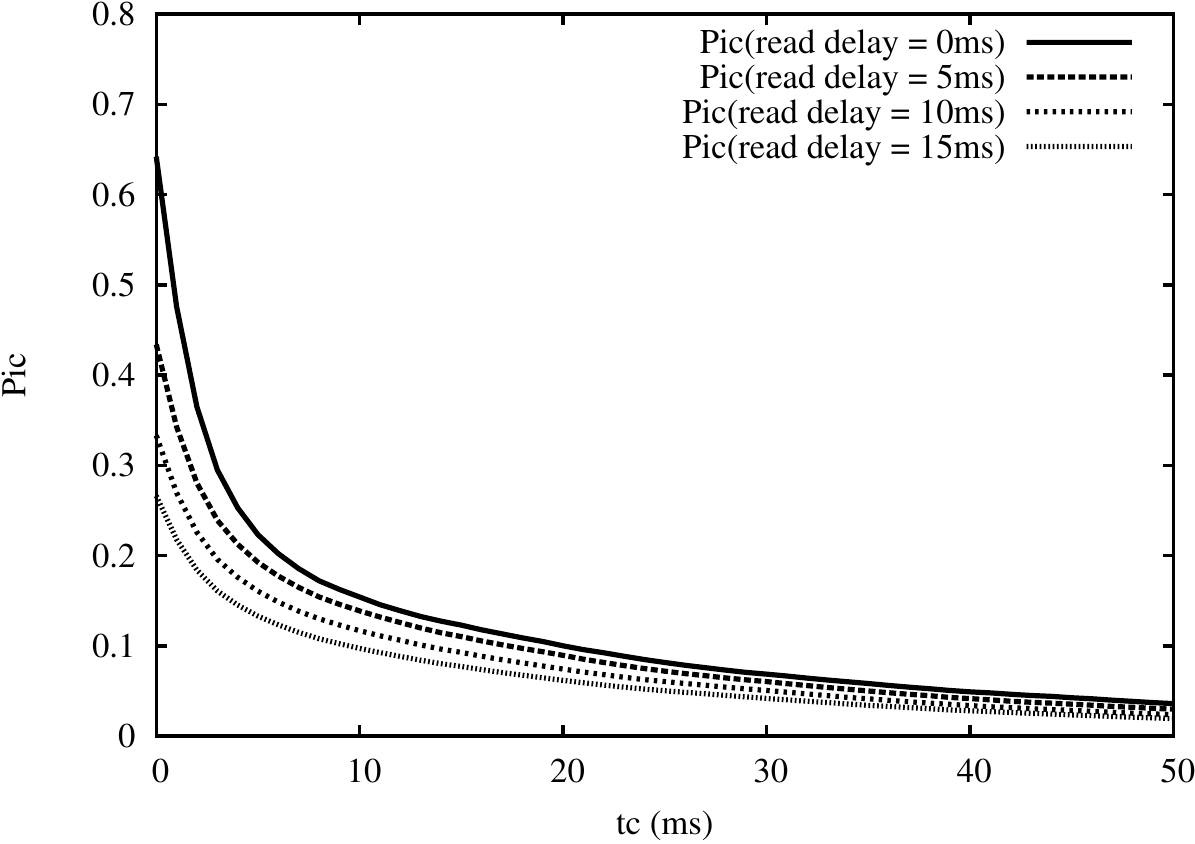}
		\vspace{-2mm}
	\caption{\it Effectiveness of Read Delay knob in PCAP Cassandra. Read repair rate fixed at 0.1.}
	\label{fig:pctcrdelay}
\end{figure}

We study the impact of two control knobs on consistency: read delay and read repair rate. 

Fig.~\ref{fig:pctcrdelay} shows the inconsistency metric $\pc$ against $t_c$ for different read delays. This shows that when applications desire fresher data (left half of the plot), read delay is flexible knob to control inconsistency $\pc$. When the freshness requirements are lax (right half of plot), the knob is less useful. However, $\pc$ is already low in this region.

On the other hand, read repair rate has a relatively smaller effect. We found that a change in read repair rate from 0.1 to 1 altered $\pc$ by only 15\%, whereas Fig.~\ref{fig:pctcrdelay} showed that a 15~ms increase in read delay (at $t_c=0~ms$) lowered inconsistency by over 50\%. As mentioned earlier, using read repair rate requires calculating $\pc$ over logs of at least $k=3000$ operations, whereas read delay worked well with $k=100$. Henceforth, by default we use read delay as our sole control knob.

\subsubsection{PCAP vs.~PBS}

\begin{figure}[htbp]
	\centering
		\includegraphics[width=0.4\textwidth]{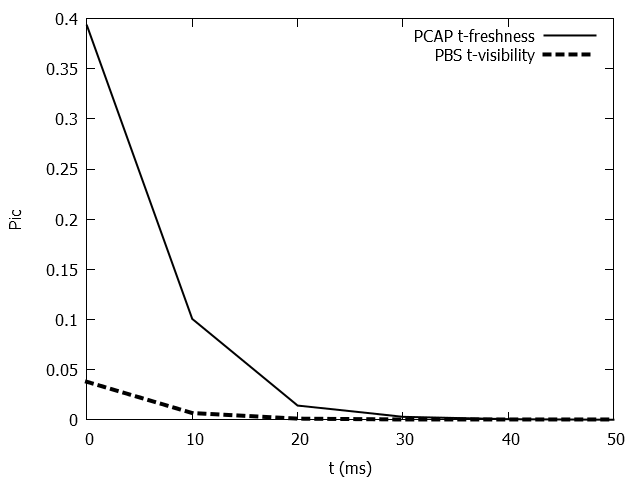}
	\caption{\it $\pc$ PCAP vs.~PBS consistency metrics. Read repair rate set to 0.1, 50\% writes.}
	\label{fig:pbs}
\end{figure}

Fig.~\ref{fig:pbs} compares, for a 50\%-write workload, the probability of inconsistency against $t$ for both existing work PBS ($t$-visibility)~\cite{pbsj} and PCAP ($t$-freshness)  described in Section~\ref{model} We observe that PBS's reported inconsistency is lower compared to PCAP. This is because, PBS considers a read that returns the value of an in-flight write (overlapping read and write) to be always fresh, by default. However the comparison between PBS and PCAP metrics is not completely fair, since the PBS metric is defined in terms of write operation end times, whereas our PCAP metric is based on write start times.  It should be noted that the purpose of this experiment is not to show which metric captures client-centric consistency better. Rather, our goal is to demonstrate that our PCAP system can be made to run by using PBS $t$-visibility metric instead of PCAP $t$-freshness.

\vspace*{-0.2cm}
\subsubsection{PCAP Metric Computation Time}
\label{sec:expt_calc}

Fig.~\ref{fig:computetime} shows the total time for the PCAP Coordinator to calculate $\pc$ and $\pa$ metrics for  values of $k$ from 100 to 10K, and using multiple threads. We observe low computation times of around 1.5 s, except when there are 64 threads and a 10K-sized log: under this situation, the system starts to degrade as too many threads contend for relatively few memory resources. Henceforth, the PCAP Coordinator by default uses a log size of $k=100$  operations and 16 threads.

\begin{figure}[htbp]
	\centering
		\includegraphics[width=0.4\textwidth]{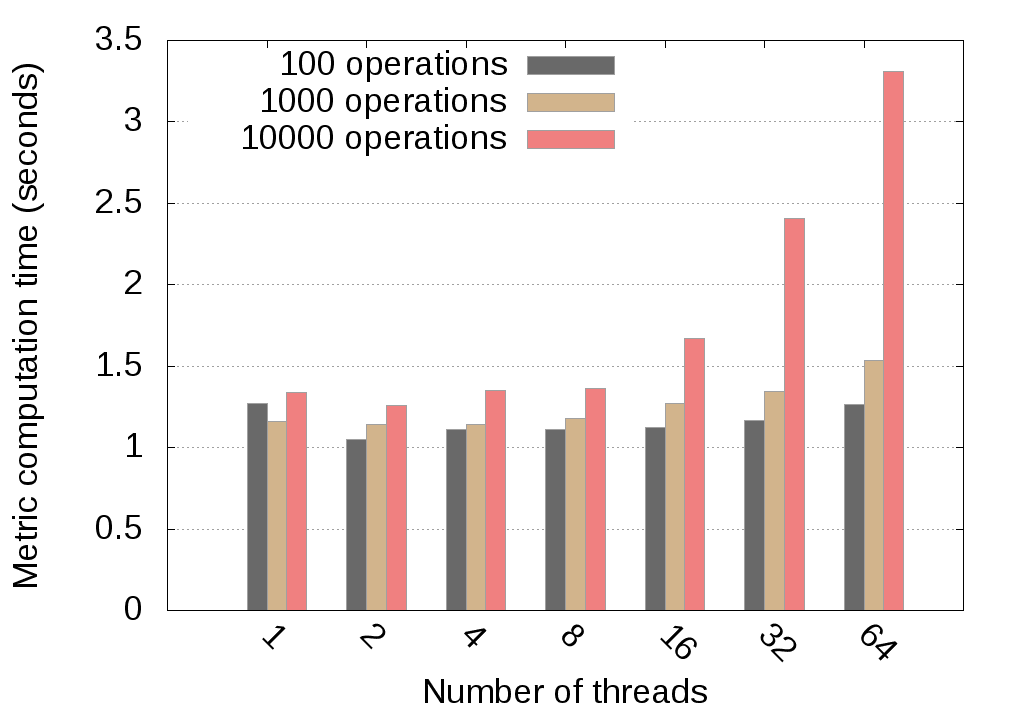}
	\caption{\it PCAP Coordinator time taken to both collect logs and compute $\pc$ and $\pa$ in PCAP Cassandra.}
	\label{fig:computetime}
\end{figure}


\subsection{Deployment Experiments}
\label{sec:expt_deployment}


We now subject our two PCAP systems to { network delay variations} and YCSB query workloads. In particular, we present two types of experiments: 1)  {\it sharp network jump} experiments, where the network delay at some of the servers changes suddenly, and 2) {\it lognormal} experiments, which inject continuously-changing and realistic delays into the network. Our experiments use $\epsilon\le0.05$ (Section~\ref{control}). 


Fig.~\ref{table:experiments} summarizes the various of SLA parameters and network conditions used in our experiments. 



\begin{figure*}[t]
{\small
\begin{center}
\begin{tabular}{|l|c|c|c|r|} \hline
System & SLA & Parameters & Delay Model & Plot\\ \hline \hline
Riak & Latency  & $\pa=0.2375$, $t_a=150~ms$, $t_c=0~ms$ & Sharp delay jump	& Fig.~\ref{fig:s_a_riak_j}\\ \hline
Riak & Consistency  & $\pc=0.17$, $t_c=0~ms$, $t_a=150~ms$ & Lognormal	& Fig.~\ref{fig:t_c_riak_ln}\\ \hline
Cassandra & Latency & $\pa=0.2375$, $t_a=150~ms$, $t_c=0~ms$ & Sharp delay jump	& Figs.~\ref{fig:t_a_cass_j},~\ref{fig:s_a_cass_j},~\ref{fig:ssc_a_cass_j}\\ \hline
Cassandra & Consistency & $\pc=0.15$, $t_c=0~ms$, $t_a=150~ms$ & Sharp delay jump	& Fig.~\ref{fig:s_c_cass_j}\\ \hline
Cassandra & Consistency & $\pc=0.135$, $t_c=0~ms$, $t_a=200~ms$ & Lognormal	& Figs.~\ref{fig:t_c_cass_ln},~\ref{fig:s_c_cass_ln},~\ref{fig:s_c_cass_ln_24},~\ref{fig:d_c_cass}\\ \hline
Cassandra & Consistency & $\pc=0.2$, $t_c=0~ms$, $t_a=200~ms$ & Lognormal	& Fig.~\ref{fig:passive}\\ \hline
Cassandra & Consistency & $\pc=0.125$, $t_c=0~ms$, $t_a=25~ms$ & Lognormal & Figs.~\ref{fig:low-timeline},~\ref{fig:low-scatter}\\
\hline
\end{tabular}
\end{center}
}
\caption{\it  {Deployment Experiments: Summary of Settings and Parameters.}}
\label{table:experiments}
\end{figure*}


\subsubsection{Latency SLA under Sharp Network Jump}
\label{sec:expt_sharp_av}

Fig.~\ref{fig:t_a_cass_j} shows the timeline of a scenario for PCAP Cassandra using the following latency SLA: $\pa^{sla}=0.2375$, $t_c=0$~ms, $t_a=150$~ms. 



In the initial segment of this run ($t=0$~s to $t=800$~s) the network delays are small; the one-way server-to-LAN switch delay is 10~ms (this is half the machine to machine delay, where a machine can be either a client or a server).  After the warm up phase, by $t=400$~s, Fig.~\ref{fig:t_a_cass_j} shows that $\pa$ has converged to the target SLA. Inconsistency $\pc$ stays close to zero. 

We wish to measure how close the PCAP system is to the optimal-achievable envelope (Section~\ref{s_pcap}). The envelope captures the lowest possible values for consistency ($\pc$, $t_c$), and latency ($\pa$, $t_a$), allowed by the network partition model ($\alpha$, $t_p$) (Theorem~\ref{thm:CAP-p}). We do this by first calculating $\alpha$ for our specific network, then calculating the optimal achievable non-SLA metric, and finally seeing how close our non-SLA metric is to this optimal. 

First, from Theorem~\ref{thm:CAP2} we know that the achievability region requires $t_c + t_a \geq t_p$; hence, we set $t_p = t_c+t_a$. Based on this, and the probability distribution of delays in the network, we calculate analytically the exact value of $\alpha$ as the fraction of client pairs whose propagation delay exceeds $t_p$ (see Definition~\ref{def:p-partition}). 

Given this value of $\alpha$ at time $t$, we can calculate the optimal value of $\pc$ as $\pc(opt)=\max(0,\alpha-\pa)$. Fig.~\ref{fig:t_a_cass_j} shows that in the initial part of the plot (until $t=800$~s), the value of $\alpha$ is close to 0, and the $\pc$ achieved by PCAP Cassandra is close to optimal.


At time $t=800$~s in Fig.~\ref{fig:t_a_cass_j}, we sharply increase the one-way server-to-LAN delay for 5 out of 9 servers from 10~ms to 26~ms. This sharp network jump results in a lossier network, as shown by the value of $\alpha$ going up from 0 to 0.42. As a result, the value of $\pa$ initially spikes -- however, the PCAP system adapts, and by time $t=1200$~s the value of $\pa$ has converged back to under the SLA. 

However, the high value of $\alpha (=0.42)$ implies that the optimal-achievable $\pc(opt)$ is also higher after $t=800~s$. Once again we notice that $\pc$ converges in the second segment of  Fig.~\ref{fig:t_a_cass_j} by $t=1200$~s. 

To visualize how close the PCAP system is to the optimal-achievable envelope, Fig.~\ref{fig:s_a_cass_j} shows the two achievable envelopes as piecewise linear segments (named ``before jump" and ``after jump'') and the $(\pa, \pc)$ data points from our run in Fig.~\ref{fig:t_a_cass_j}. The figure annotates the clusters of data points by their time interval. We observe that in the stable states both before the jump (dark circles) and after the jump (empty triangles) are close to their optimal-achievable envelopes.

\begin{figure}[!htb]
    \centering
    \begin{minipage}{.5\textwidth}
        \centering
        \includegraphics[width=0.9\linewidth, height=0.25\textheight]{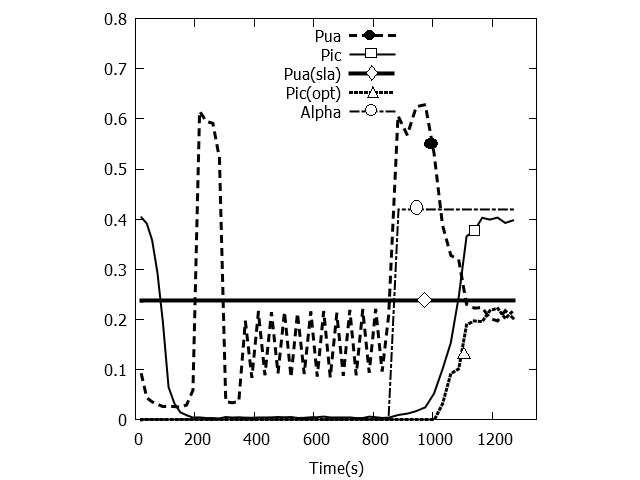}
        \caption{\it Latency SLA with PCAP Cassandra \\ under Sharp Network Jump at 800 s: Timeline.}
        \label{fig:t_a_cass_j}
    \end{minipage}%
    \begin{minipage}{0.5\textwidth}
        \centering
        \includegraphics[width=0.9\linewidth, height=0.25\textheight]{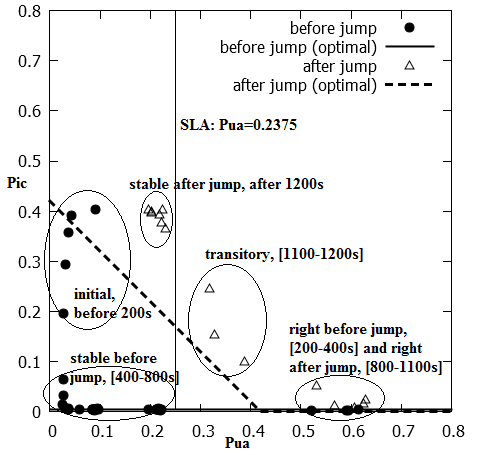}
        \caption{\it Latency SLA with PCAP Cassandra \\ under Sharp Network Jump: Consistency-Latency Scatter plot.}
        \label{fig:s_a_cass_j}
    \end{minipage}
\end{figure}


\begin{figure}[htbp]
	\centering
		\includegraphics[width=0.45\textwidth]{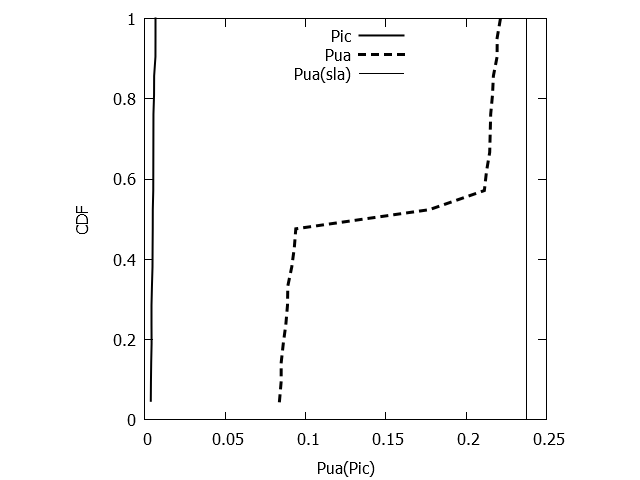}
	\caption{\it Latency SLA with PCAP Cassandra under Sharp Network Jump: Steady State CDF [400~s, 800~s].}
	\label{fig:ssc_a_cass_j}
\end{figure}

Fig.~\ref{fig:ssc_a_cass_j} shows the CDF plot for $\pa$ and $\pc$ in the steady state time interval [400~s, 800~s] of Fig.~\ref{fig:t_a_cass_j}, corresponding to the bottom left cluster from Fig.~\ref{fig:s_a_cass_j}. We observe that $\pa$ is always below the SLA.

Fig.~\ref{fig:s_a_riak_j} shows a scatter plot for our PCAP Riak system under a latency SLA ($\pa^{sla}=0.2375$, $t_a=150~ms$, $t_c=0~ms$).   The sharp network jump occurs at time $t=4300$~s when we increase the one-way server-to-LAN delay for 4 out of the 9 Riak nodes from 10~ms to 26~ms. It takes about 1200~s for $\pa$ to converge to the SLA (at around $t=1400$~s in the warm up segment and $t=5500$~s in the second segment). 



\subsubsection{Consistency SLA under Sharp Network Jump}
\label{sec:expt_sharp_cons}


\begin{figure}[!htb]
    \centering
    \begin{minipage}{.5\textwidth}
        \centering
        \includegraphics[width=0.9\linewidth, height=0.25\textheight]{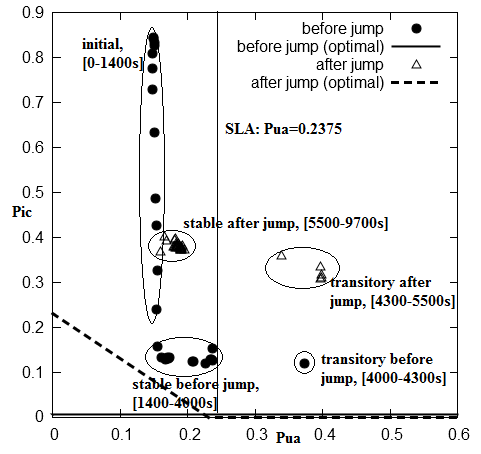}
        \caption{\it Latency SLA with PCAP \\Riak under Sharp Network Jump: \\Consistency-Latency Scatter plot..}
        \label{fig:s_a_riak_j}
    \end{minipage}%
    \begin{minipage}{0.5\textwidth}
        \centering
        \includegraphics[width=0.9\linewidth, height=0.25\textheight]{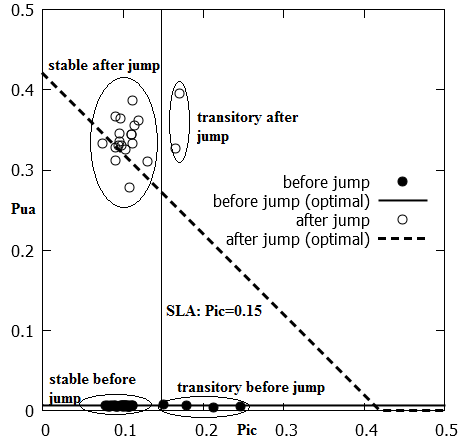}
        \caption{\it Consistency SLA with PCAP \\Cassandra under Sharp Network Jump: \\Consistency-Latency Scatter plot.}
        \label{fig:s_c_cass_j}
    \end{minipage}
\end{figure}


We present consistency SLA results for PCAP Cassandra (PCAP Riak results are similar and are omitted). We use $\pc^{sla}=0.15$, $t_c=0~ms$, $t_a=150~ms$. The initial  one-way server-to-LAN delay is 10~ms. At time 750~s, we increase the one-way server-to-LAN delay for 5 out of 9 nodes to 14~ms. This changes $\alpha$ from 0 to 0.42. 

Fig.~\ref{fig:s_c_cass_j} shows the scatter plot. First, observe that the PCAP system meets the consistency SLA requirements, both before and after the jump. Second, as network conditions worsen, the optimal-achievable envelope moves significantly. Yet the PCAP system remains close to the optimal-achievable envelope. The convergence time is about 100~s, both before and after the jump.

\vspace*{-0.2cm}
\subsubsection{Experiments with Realistic Delay Distributions}
\label{sec:expt_lognormal}


This section evaluates the behavior of PCAP Cassandra and PCAP Riak under continuously-changing network conditions and a consistency SLA (latency SLA experiments yielded similar results and are omitted). 


\begin{figure}[!htb]
    \centering
    \begin{minipage}{.5\textwidth}
        \centering
        \includegraphics[width=0.9\linewidth, height=0.25\textheight]{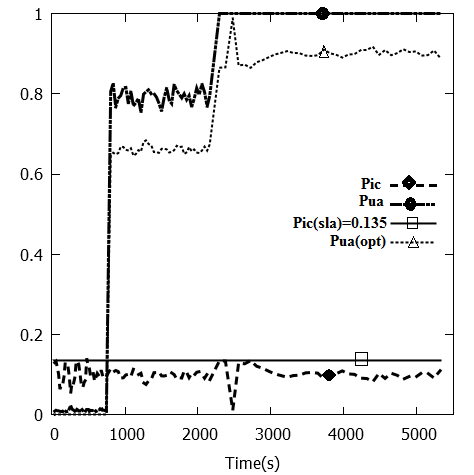}
        \caption{\it Consistency SLA with PCAP \\Cassandra under Lognormal delay distribution:\\Timeline.}
        \label{fig:t_c_cass_ln}
    \end{minipage}%
    \begin{minipage}{0.5\textwidth}
        \centering
        \includegraphics[width=0.9\linewidth, height=0.25\textheight]{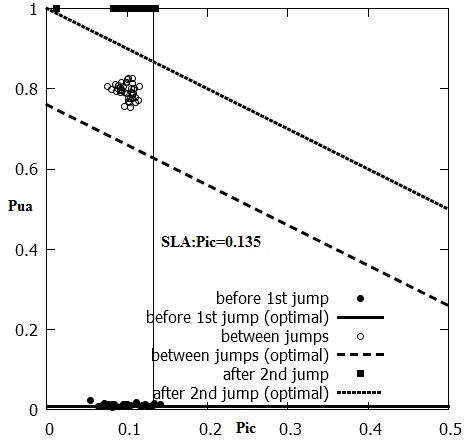}
        \caption{\it Consistency SLA with PCAP \\Cassandra under Lognormal delay distribution:\\Consistency-Latency Scatter plot.}
        \label{fig:s_c_cass_ln}
    \end{minipage}
\end{figure}

Based on studies for enterprise data-centers~\cite{BAM10} we use a lognormal distribution for injecting packet delays into the network. We modified the Linux traffic shaper to add lognormally distributed delays to each packet. Fig.~\ref{fig:t_c_cass_ln} shows a timeline where initially ($t=0$ to 800~s) the delays are lognormally distributed, with the underlying normal distributions of $\mu=3$~ms and $\sigma=0.3$~ms. At $t=800$~s we increase $\mu$ and $\sigma$ to $4$~ms and $0.4$~ms respectively. Finally at around 2100~s, $\mu$ and $\sigma$ become $5$~ms and $0.5$~ms respectively. Fig.~\ref{fig:s_c_cass_ln} shows the corresponding scatter plot. We observe that in all three time segments, the inconsistency metric $\pc$: i) stays below the SLA, and ii) upon a sudden network change converges back to the SLA. Additionally, we observe that $\pa$ converges close to its optimal achievable value. 

Fig.~\ref{fig:t_c_riak_ln} shows the effect of worsening network conditions on PCAP Riak. At around $t=1300$~s we increase $\mu$ from 1~ms to 4~ms, and $\sigma$ from 0.1~ms to 0.5~ms. The plot shows that it takes PCAP Riak an additional 1300~s to have inconsistency $\pc$ converge to the SLA. Further the non-SLA metric $\pa$ converges close to the optimal.  
\begin{figure}[htbp]
	\centering
		\includegraphics[width=0.5\textwidth]{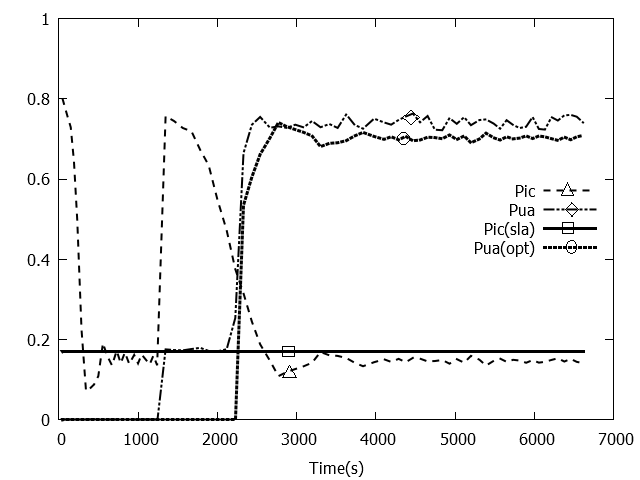}
	\caption{\it Consistency SLA with PCAP Riak under Lognormal delay distribution: Timeline.}
	\label{fig:t_c_riak_ln}
\end{figure}

So far all of our experiments used lax timeliness requirements ($t_a=150~ms,200~ms$), and were run on top of relatively high delay networks. Next we perform a stringent consistency SLA experiment ($t_c=0~ms,\pc=.125$) with a very tight latency timeliness requirement ($t_a=25~ms$). Packet delays are still lognormally distributed, but with lower values. Fig.~\ref{fig:low-timeline} shows a timeline where initially the delays are lognormally distributed with $\mu=1~ms$, $\sigma=0.1~ms$. At time $t=160~s$ we increase $\mu$ and $\sigma$ to $1.5~ms$ and $0.15~ms$ respectively. Then at time $t=320~s$, we decrease $\mu$ and $\sigma$ to return to the initial network conditions. We observe that in all three time segments, $\pc$ stays below the SLA, and quickly converges back to the SLA after a network change. Since the network delays are very low throughout the experiment, $\alpha$ is always 0. Thus the optimal $\pa$ is also 0. We observe that $\pa$ converges very close to optimal before the first jump and after the second jump ($\mu=1~ms,\sigma=0.1~ms$). In the middle time segment ($t=160$ to $320~s$), $\pa$ degrades in order to meet the consistency SLA under slightly higher packet delays. Fig.~\ref{fig:low-scatter} shows the corresponding scatter plot. We observe that the system is close to the optimal envelope in the first and last time segments, and the SLA is always met. We note that we are far from optimal in the middle time segment, when the network delays are slightly higher. This shows that when the network conditions are relatively good, the PCAP system is close to the optimal envelope, but when situations worsen we move away. The gap between the system performance and the envelope indicates that the bound (Theorem~\ref{thm:CAP-p}) could be improved further. We leave this as an open question.


\begin{figure}[!htb]
    \centering
    \begin{minipage}{.5\textwidth}
        \centering
        \includegraphics[width=0.9\linewidth, height=0.3\textheight]{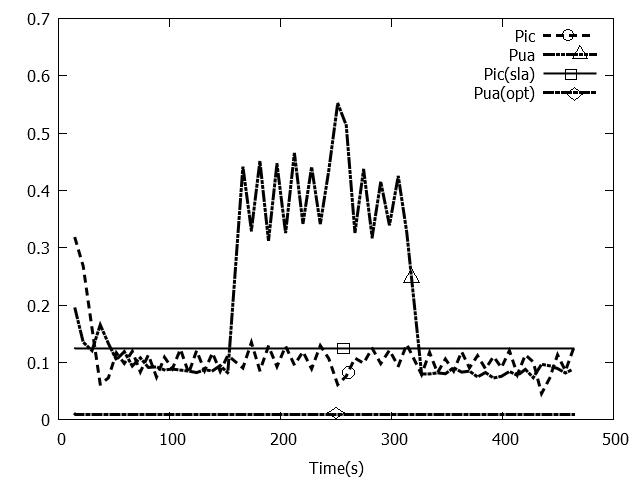}
        \caption{\it Consistency SLA with PCAP \\Cassandra under Lognormal delay: \\Timeline ($t_c=0~ms,\pc=0.125,t_a=25~ms$).}
        \label{fig:low-timeline}
    \end{minipage}%
    \begin{minipage}{0.5\textwidth}
        \centering
        \includegraphics[width=0.9\linewidth, height=0.3\textheight]{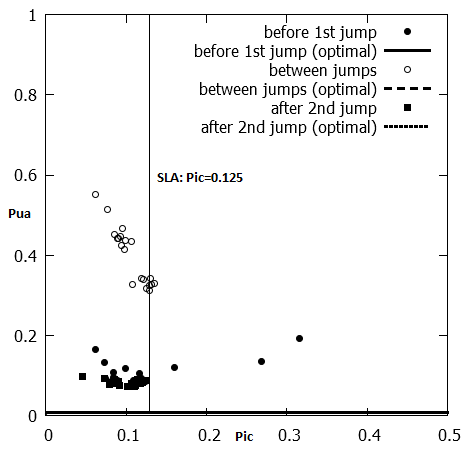}
        \caption{\it Consistency SLA with PCAP \\Cassandra under Lognormal delay: \\Scatter Plot ($t_c=0~ms,\pc=0.125,t_a=25~ms$).}
        \label{fig:low-scatter}
    \end{minipage}
\end{figure}


\subsubsection{Effect of Read Repair Rate Knob}
\vspace*{-0.1cm}
\label{sec:expt_readrepair}
\begin{figure}[htbp]
	\centering
		\includegraphics[width=0.40\textwidth]{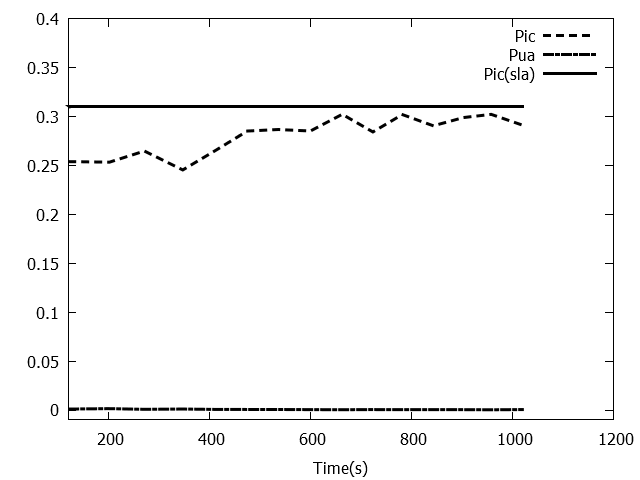}
	\caption{\it Effect of Read Repair Rate on PCAP Cassandra. $\pc=0.31$, $t_c=0~ms$, $t_a=100~ms$.}
	\label{fig:t_c_cass_rr}
\end{figure}

All of our deployment experiments  use read delay as the only control knob. Fig.~\ref{fig:t_c_cass_rr} shows a portion of a run when only read repair rate was used by our PCAP Cassandra system. This was because read delay was already zero, and we needed to push $\pc$ up to $\pc^{sla}$. First we notice that $\pa$ does not change with read repair rate, as expected (Table.~\ref{table:knobs}). Second, we notice that the convergence of $\pc$ is very slow -- it changes from 0.25 to 0.3 over a long period of 1000~s. 

Due to this slow convergence, we conclude that read repair rate is useful only when network delays remain relatively stable. Under continuously changing network conditions (e.g., a lognormal distribution) convergence may be slower and thus read delay should be used as the only control knob.

\subsubsection{Scalability}
\vspace{-1mm}
\begin{figure}[htbp]
	\centering
		\includegraphics[width=0.4\textwidth]{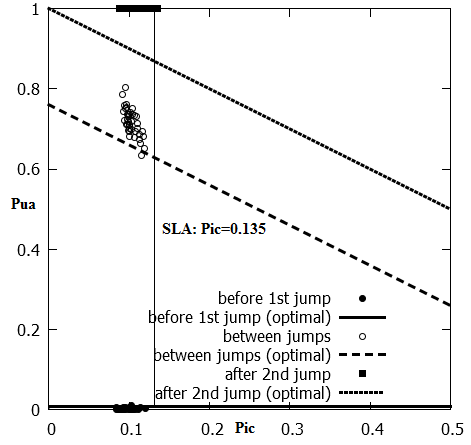}
	\caption{\it Scatter plot for same settings as Fig.~\ref{fig:s_c_cass_ln}, but with 32 servers and 16K ops/s.}
	\label{fig:s_c_cass_ln_24}
\end{figure}

We measure scalability via an increased workload on PCAP Cassandra. Compared to Fig.~\ref{fig:s_c_cass_ln}, in this new run we increased the number of servers from 9 to 32, and throughput to 16000 ops/s, and ensured that each server stores at least some keys. All other settings are unchanged compared to Fig.~\ref{fig:s_c_cass_ln}. The result is shown Fig.~\ref{fig:s_c_cass_ln_24}. Compared with Fig.~\ref{fig:s_c_cass_ln}, we observe an improvement with scale -- in particular, increasing the number of servers brings the system closer to optimal. 

\subsubsection{Effect of Timeliness Requirement}
\label{sec:exp_ttimeliness}
The timeliness requirements in an SLA directly affect how close the PCAP system is to the optimal-achievable envelope. Fig.~\ref{fig:d_c_cass} shows the effect of varying the timeliness parameter $t_a$ in a consistency SLA ($t_c=0~ms$, $\pc=0.135$) experiment for PCAP Cassandra with 10~ms node to LAN delays. For each $t_a$, we consider the cluster of the  $(\pa,\pc)$ points achieved by the PCAP system in its stable state, calculate its centroid, and measure (and plot on vertical axis) the distance $d$ from this centroid to the optimal-achievable consistency-latency envelope. Note that the optimal envelope calculation also involves $t_a$, since $\alpha$ depends on it (Section~\ref{sec:expt_sharp_av}). 

Fig.~\ref{fig:d_c_cass} shows that when $t_a$ is too stringent ($<$ 100~ms), the PCAP system may be far from the optimal envelope even when it satisfies the SLA. In the case of Fig.~\ref{fig:d_c_cass}, this is because in our network, the average time to cross four hops (client to coordinator to replica, and the reverse) is $20 \times 4 = 80$~ms.\footnote{{ Round-trip time for each hop is $2 \times 10 = 20$~ms.}} As $t_a$ starts to go beyond this (e.g., $t_a \geq$ 100~ms), the timeliness requirements are less stringent,and PCAP is essentially optimal (very close to the achievable envelope).
\begin{figure}[htbp]
	\centering
		\includegraphics[width=0.45\textwidth]{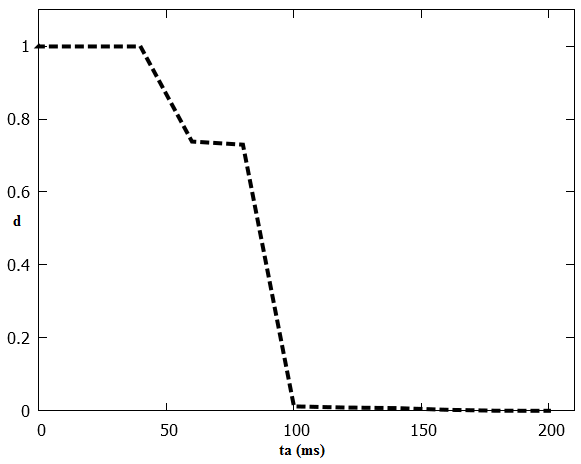}
	\caption{\it Effect of Timeliness Requirement ($t_a$) on PCAP Cassandra. Consistency SLA with $\pc=0.135$, $t_c=0~ms$.}
	\label{fig:d_c_cass}
\end{figure}

\subsubsection{Passive Measurement Approach}
\label{sec:passive}

So far all our experiments have used the active measurement approach. In this section, we repeat a PCAP Cassandra consistency SLA experiment ($\pc=0.2$, $t_c=0~ms$) using a passive measurement approach. 

In Figure~\ref{fig:passive}, instead of actively injecting operations, we sample ongoing client operations. We estimate $\pc$ and $\pa$ from the 100 latest operations from 5 servers selected randomly. 


At the beginning, the delay is lognormally distributed with $\mu=1~ms$, $\sigma=0.1~ms$. The passive approach initially converges to the SLA. We change the delay ($\mu=2~ms$, $\sigma=0.2~ms$) at $t=325~s$. We observe that, compared to the active approach, 1) consistency (SLA metric) oscillates more, and 2) the availability (non-SLA metric) is farther from optimal and takes longer to converge. For the passive approach, SLA convergence and non-SLA optimization depends heavily on the sampling of operations used to estimate the metrics. Thus we conclude that it is harder to satisfy SLA and optimize the non-SLA metric with the passive approach.


\begin{figure}[htbp]
	\centering
		\includegraphics[width=0.4\textwidth]{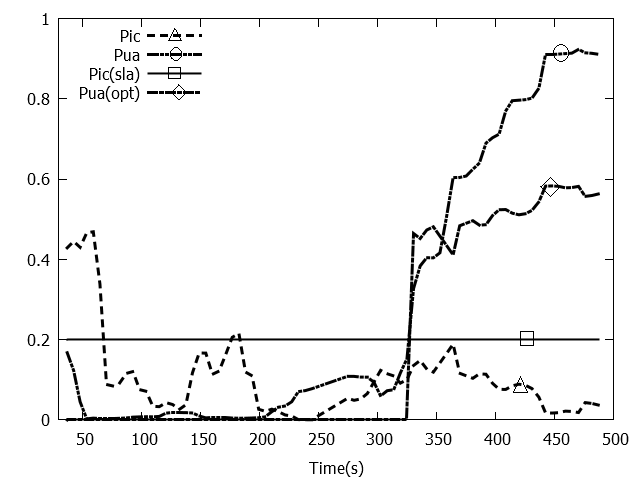}
	\caption{\it Consistency SLA with PCAP Cassandra under Lognormal delay distribution: Timeline (Passive).}
	\label{fig:passive}
\end{figure}

\subsection{GeoPCAP Evaluation}
\label{sec:geo-pcap-eval}

We evaluate GeoPCAP with a Monte-Carlo simulation. In our setup, we have four data-centers, among which three are remote data-centers holding replicas of a key, and the fourth one is the local data-center. At each iteration, we estimate $t$-freshness per data-center using a variation of the well-known WARS model~\cite{pbsj}. The WARS model is based on Dynamo style quorum systems~\cite{DHJ07}, where data staleness is due to read and write message reordering. The model has four components. $W$ represents the message delay from coordinator to replica. The acknowledgment from the replica back to the coordinator is modeled by a random variable $A$. The read message delay from coordinator to replica, and the acknowledgment back are represented by $R$, and $S$, respectively. A read will be stale if a read is acknowledged before a write reaches the replica, i.~e.~, $R+S < W+A$. In our simulation, we ignore the $A$ component since we do not need to wait for a write to finish before a read starts. We use the LinkedIn SSD disk latency distribution~\cite{pbsj}, Table 3 for read/write operation latency values. 

We model the WAN delay using a normal distribution $N(20~ms,\sqrt{2}~ms)$ based on results from~\cite{WAN-model}. Each simulation runs for 300 iterations. At each iteration, we run the PID control loop (Figure~\ref{alg:gadapt}) to estimate a new value for geo-delay $\Delta$, and sleep for 1 sec. All reads in the following iteration are delayed at the local data-center by $\Delta$.  At iteration 150, we inject a jump by increasing the mean and standard deviation of each WAN link delay normal distribution to $22~ms$ and $\sqrt{2.2}~ms$, respectively. We show only results for consistency and latency SLA for the \texttt{ALL} composition. The \texttt{QUICKEST} composition results are similar and are omitted. 

\begin{figure}[!htb]
    \centering
    \begin{minipage}{.5\textwidth}
        \centering
        \includegraphics[width=0.9\linewidth, height=0.25\textheight]{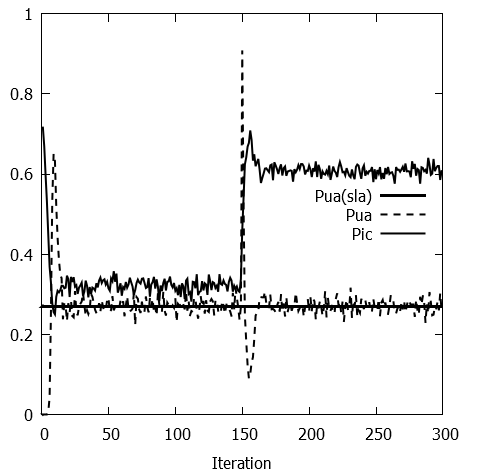}
        \caption{\it GeoPCAP SLA Timeline \\for L SLA (\texttt{ALL}).}
        \label{fig:geo-av-all-sla}
    \end{minipage}%
    \begin{minipage}{0.5\textwidth}
        \centering
        \includegraphics[width=0.9\linewidth, height=0.25\textheight]{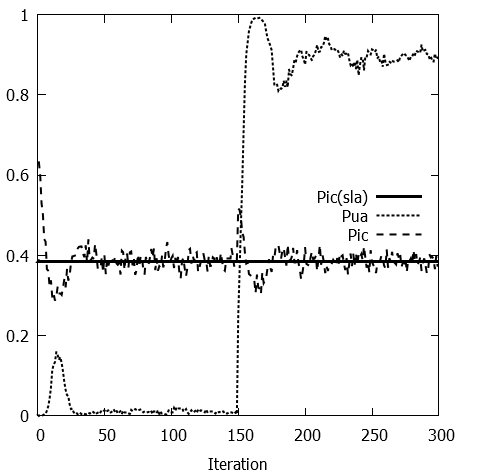}
        \caption{\it GeoPCAP SLA Timeline \\for C SLA (\texttt{ALL}).}
        \label{fig:geo-con-all-sla}
    \end{minipage}
\end{figure}

Figure~\ref{fig:geo-av-all-sla} shows the timeline of SLA convergence for GeoPCAP Latency SLA ($\pa^{sla}=0.27, ta^{sla}=25~ms, tc^{sla}=0.1~ms$). We observe that using the PID controller ($k_p=1$, $k_d=0.5$, $k_i=0.5$), both the SLA and the other metric converge within 5 iterations initially and also after the jump. Figure~\ref{fig:geo-av-all-rdelay} shows the corresponding evolution of the geo-delay control knob. Before the jump, the read delay converges to around 5~ms. After the jump, the WAN delay increase forces the geo-delay to converge to a lower value (around 3~ms) in order to meet the latency SLA.

\begin{figure}[!htb]
    \centering
    \begin{minipage}{.5\textwidth}
        \centering
        \includegraphics[width=0.9\linewidth, height=0.25\textheight]{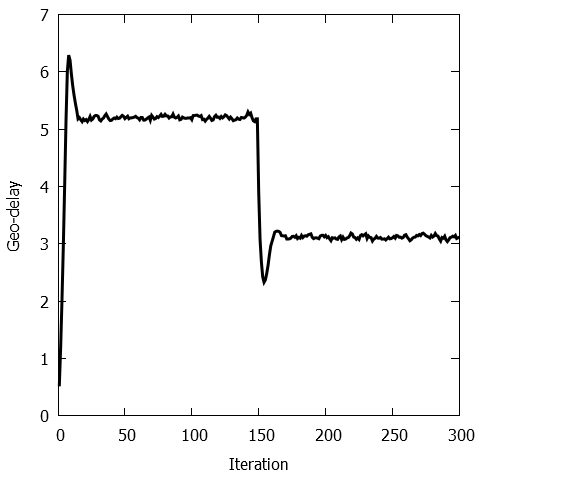}
        \caption{\it Geo-delay Timeline \\for L SLA (\texttt{ALL}) (Figure~\ref{fig:geo-av-all-sla}).}
        \label{fig:geo-av-all-rdelay}
        
    \end{minipage}%
    \begin{minipage}{0.5\textwidth}
        \centering
        \includegraphics[width=0.9\linewidth, height=0.25\textheight]{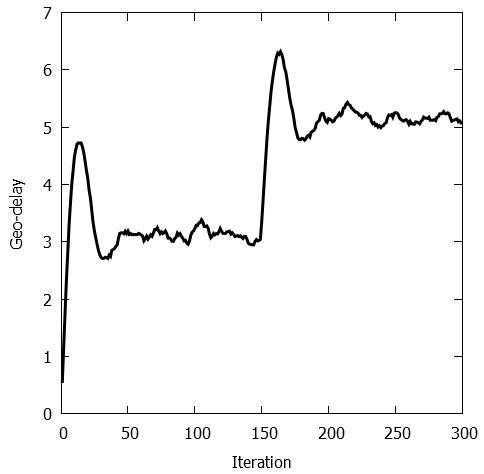}
        \caption{\it Geo-delay Timeline \\for C SLA (\texttt{ALL}) (Figure~\ref{fig:geo-con-all-sla}).}
        \label{fig:geo-con-all-rdelay}
    \end{minipage}
\end{figure}

Figure~\ref{fig:geo-con-all-sla} shows the consistency SLA ($\pc^{sla}=0.38, tc^{sla}=1~ms, ta^{sla}=25~ms$) time line. Here convergence takes 25 iterations, thanks to the PID controller (($k_p=1$, $k_d=0.8$, $k_i=0.5$)). We needed a slightly higher value for the differential gain $k_d$ to deal with increased oscillation for the consistency SLA experiment. Note that the $\pc^{sla}$ value of 0.38 forces a smaller per data-center $\pc$ convergence. The corresponding geo-delay evolution (Figure~\ref{fig:geo-con-all-rdelay}) initially converges to around 3~ms before the jump, and converges to around 5~ms after the jump, to enforce the consistency SLA after the delay increase.  



\begin{figure}[htbp]
	\centering
		\includegraphics[width=0.4\textwidth]{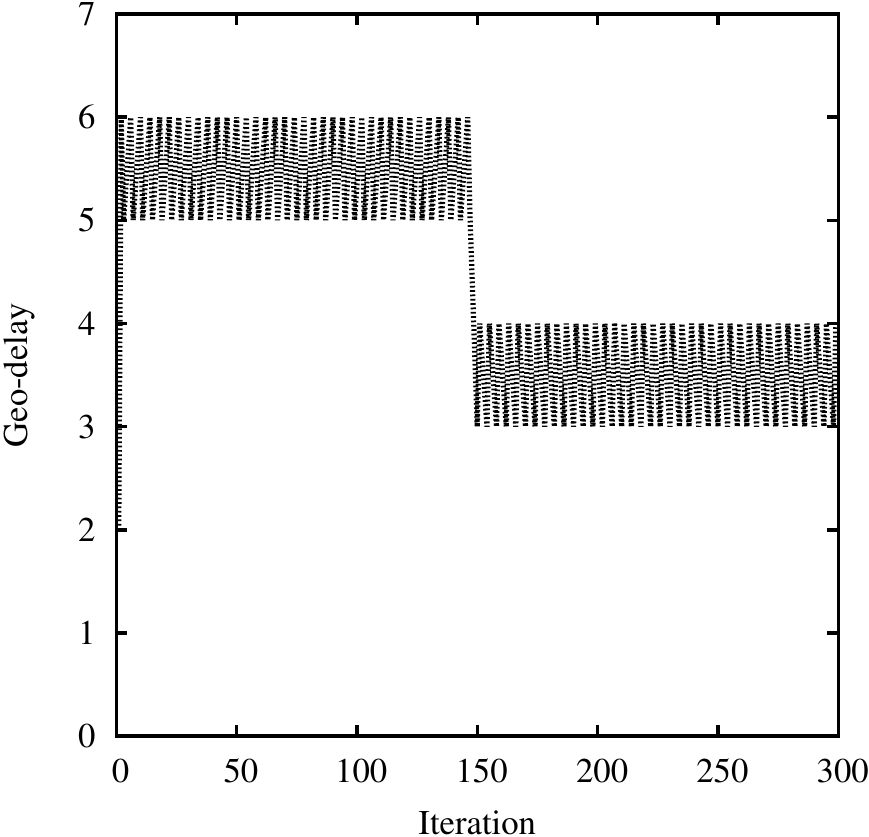}
	\caption{\it Geo-delay Timeline for A SLA (\texttt{ALL}) with Multiplicative Approach.}
	\label{fig:geo-av-all-rdelay-mult}
\end{figure}

We also repeated the Latency SLA experiment with the \texttt{ALL} composition (Figures~\ref{fig:geo-av-all-sla},~\ref{fig:geo-av-all-rdelay}) using the multiplicative control approach (Section~\ref{control}) instead of the PID control approach. Figure~\ref{fig:geo-av-all-rdelay-mult} shows the corresponding geo-delay trend compared to Figure~\ref{fig:geo-av-all-rdelay}. Comparing the two figures, we observe that although the multiplicative strategy converges as fast the PID approach both before and after the delay jump, the read delay value keeps oscillating around the optimal value. Such oscillations cannot be avoided in the multiplicative approach, since at steady state the control loop keeps changing direction with a unit step size. Compared to the multiplicative approach, the PID control approach is smoother and has less oscillations.




\section{Related Work}
\label{sec:related}




\subsection{Consistency-Latency Tradeoffs} 
There has been work on theoretically characterizing the tradeoff between latency and strong consistency models. Attiya and Welch~\cite{attiya_welch_TOCS94} studied the tradeoff between latency and linearizability and sequential consistency. Subsequent work has explored linearizablity under different delay models~\cite{EM_disc99,Roth_TCS}. All these papers are concerned with strong consistency models whereas we consider $t$-freshness, which models data freshness in eventually consistent systems. Moreover, their delay models are different from our partition model. There has been theoretical work on probabilistic quorum systems~\cite{Malkhi:1997:PQS:259380.259458,DBLP:conf/icdcs/LeeW01,DBLP:journals/dc/AbrahamM05}. Their consistency models are different from ours; moreover, they did not consider the tradeoff between consistency and availability.

There are two classes of systems that are closest to our work. The first class of systems are concerned with metrics for measuring data freshness or staleness. We do not compare our work against this class of systems in this paper, as it is not our goal to propose yet another consistency model or metric. Bailis et al.~\cite{pbsw,pbsj} propose a probabilistic consistency model (PBS) for quorum-based stores, but did not consider latency, soft partitions or the CAP theorem. Golab et al.~\cite{gls:fun} propose a time-based staleness metric called $\Delta$-atomicity. $\Delta$-atomicity is considered the gold standard for measuring atomicity violations (staleness) across multiple read and write operations. The $\Gamma$ metric~\cite{gamma} is inspired by the $\Delta$ metric and improves upon it on multiple fronts. For example, the $\Gamma$ metric makes fewer technical assumptions than the $\Delta$ metric and produces less noisy results. It is also more robust against clock skew. All these related data freshness metrics cannot be directly compared to our $t$-freshness metric. The reason is that unlike our metric which considers write start times, these existing metrics consider end time of write operations when calculating data freshness.  

The second class of systems deal with adaptive mechanisms for meeting consistency-latency SLAs for key-value stores. The Pileus system~\cite{pileus} considers families of consistency/latency SLAs, and requires the application to specify a utility value with each SLA. In comparison, PCAP considers probabilistic metrics of $\pc,\pa$. Tuba~\cite{tuba} extends the predefined and static Pileus mechanisms with dynamic replica reconfiguration mechanisms to maximize Pileus style utility functions without impacting client read and write operations. Golab and Wylie~\cite{GW2014} propose consistency amplification, which is a framework for supporting consistency SLAs by injecting delays at servers or clients. In comparison, in our PCAP system, we only add delays at servers. McKenzie et al.~\cite{GolabCPQ15} propose continuous partial quorums (CPQ), which is a technique to randomly choose between multiple discrete consistency levels for fine-grained consistency-latency tuning, and compare CPQ against consistency amplification. Compared to all these systems where the goal is to meet SLAs, in our work, we also (1) quantitatively characterize the (un)achievable consistency-latency tradeoff envelope, and (2) show how to design systems that perform close to this envelope, in addition to (3) meeting SLAs. The PCAP system can be setup to work with any of these SLAs listed above; but we don't do this in the paper since our main goal is to measure how close the PCAP system is to the optimal consistency-latency envelope.  

Recently, there has been work on declarative ways to specify application consistency and latency requirements -- PCAP proposes mechanisms to satisfy such specifications~\cite{DBLP:conf/pldi/Sivaramakrishnan15}.






\subsection{Adaptive Systems}
There are a few existing systems that controls consistency in storage systems. {FRACS~\cite{zz:trading} controls consistency by allowing replicas to buffer
updates up to a given staleness. 
AQuA~\cite{ksc:qos} continuously moves replicas between ``strong'' and ``weak'' consistency groups to implement different consistency levels. Fox and Brewer~\cite{harvestyield} show how to trade consistency (harvest) for availability (yield) in the context of the Inktomi search engine. While harvest and yield capture continuously changing consistency and availability conditions, we characterize the consistency-availability (latency) tradeoff in a quantitative manner. 
TACT~\cite{conit} controls staleness by limiting the number of outstanding writes at replicas (order error) and bounding write propagation delay (staleness).} All the mentioned systems provide {\it best-effort} behavior for consistency, within the latency bounds. In comparison, the PCAP system explicitly allows applications to specify SLAs. Consistency levels have been adaptively changed to deal with node failures and network changes in~\cite{pbsf}, however this may be intrusive for applications that explicitly set consistency levels for operations. Artificially delaying read operations at servers (similar to our read delay knob) has been used to eliminate staleness spikes (improve consistency) which are correlated with garbage collection in a specific key-value store (Apache Cassandra)~\cite{Fan:2015:UCC:2752939.2752949}. Similar techniques have been used to guarantee causal consistency for client-side applications~\cite{Zawirski:2015:WFR:2814576.2814733}. Simba~\cite{simba} proposes new consistency abstractions for mobile application data synchronization services, and allows applications to choose among various consistency models.

For stream processing, Gedik et al.~\cite{6678504} propose a control algorithm to compute the optimal resource requirements to meet throughput requirements. There has been work on adaptive elasticity control for storage~\cite{Lim:2010:ACE:1809049.1809051}, and adaptively tuning Hadoop clusters to meet  SLAs~\cite{Herodotou:2011:NOS:2038916.2038934}. Compared to the controllers present in these systems, our PCAP controller achieves control objectives~\cite{Hellerstein:2004:FCC:975344} using a different set of techniques to meet SLAs for key-value stores.   

\subsection{Composition} Composing local policies to form global policies is well studied in other domains, for example QoS composition in multimedia networks~\cite{DBLP:journals/mms/LiangN06}, software defined network (sdn) composition~\cite{Monsanto:2013:CSN:2482626.2482629}, and web-service orchestration~\cite{Dustdar:2005:SWS:1358537.1358538}. Our composition techniques are aimed at consistency and latency guarantees for geo-distributed systems.
\section{Summary}
In this paper, we have first formulated and proved a probabilistic variation of the CAP theorem which took into account probabilistic models for consistency, latency, and soft partitions within a data-center. Our theorems show the un-achievable envelope, i.e., which combinations of these three models make them impossible to achieve together. We then show how to design systems (called PCAP) that (1) perform close to this optimal envelope, and (2) can meet consistency and latency SLAs derived from the corresponding models. 
We then incorporated these SLAs into Apache Cassandra and Riak running in a single data-center. We also extended our PCAP system from a single data-center to multiple geo-distributed data-centers. Our experiments with YCSB workloads and realistic traffic demonstrated that our PCAP system meets the SLAs, that its performance is close to the optimal-achievable consistency-availability envelope, and that it scales well. Simulations of our GeoPCAP system also showed SLA satisfaction for applications spanning multiple data-centers.

   



\bibliographystyle{abbrv}
\bibliography{bibfile}

\end{document}